\documentclass[11pt,onecolumn]{article}
\usepackage[utf8]{inputenc}
\usepackage{amssymb}
\usepackage{amsthm}
\usepackage{thm-restate}
\usepackage{amsmath}
\usepackage{graphicx}
\usepackage{color,soul}
\usepackage{xcolor}
\usepackage{empheq}
\usepackage[many]{tcolorbox}
\usepackage{mathtools}
\usepackage{subfig}
\usepackage{stmaryrd}
\usepackage{enumitem}
\usepackage[vlined,linesnumbered,ruled,resetcount]{algorithm2e}

\usepackage[left=1.0in,top=1.0in,right=1.0in,bottom=1.0in,nohead,footskip=0.3in]{geometry}

\theoremstyle{plain}
\newtheorem{theorem}{Theorem}[section]

\newtheorem{definition}[theorem]{Definition}
\newtheorem{lemma}[theorem]{Lemma}

\newtheorem{remark}[theorem]{Remark}

\newenvironment{myproblem}[2]{\vspace{5pt}
\noindent {\bf Problem #1} (#2). \em 
}{\vspace{5pt} \\ \indent }

\let\oldnl\nl% Store \nl in \oldnl
\newcommand{\nonl}{\renewcommand{\nl}{\let\nl\oldnl}}% Remove line number for one line
\newcommand{\eqdef}{{\stackrel{\mbox{\tiny \tt ~def~}}{=}}}

\SetKwProg{Procedure}{Procedure}{}{}
\SetKwRepeat{Do}{do}{while}

\usepackage{hyperref}

\makeatletter
\newcommand{\customlabel}[2]{%
   \protected@write \@auxout {}{\string \newlabel {#1}{{#2}{\thepage}{#2}{#1}{}} }%
   \hypertarget{#1}{#2}
}
\makeatother

\long\def\ignore#1{}

\def\Def{{\tt{def}}}
\def\mdef{{\widetilde{\tt{def}}}}
\def\indeg{{\tt{indeg}}}
\def\outdeg{{\tt{outdeg}}}

\def\head{{\tt{head}}}
\def\tail{{\tt{tail}}}

\def\calS{{\cal S}}

\def\calM{{\cal M}}

\def\calG{{\cal G}}

\def\calD{{\cal D}}
\def\calX{{\cal X}}

\newcommand{\C}[1]{{\langle #1 \rangle}}
\newcommand{\CC}[1]{{\langle\!\langle #1 \rangle\!\rangle}}
\newcommand{\CCd}[1]{{\langle\!\langle #1 \rangle\!\rangle_d}}
\newcommand{\CCzero}[1]{{\langle\!\langle #1 \rangle\!\rangle_0}}
\newcommand{\CCone}[1]{{\langle\!\langle #1 \rangle\!\rangle_1}}

\title{Faster algorithms for packing forests in graphs \\ and related problems}
\author{Pavel Arkhipov \hspace{30pt} Vladimir Kolmogorov \\ \normalsize Institute of Science and Technology Austria \\ {\normalsize\tt $\{$pavel.arkhipov,vnk$\}$@ist.ac.at}}
\date{}

\begin{document}

\maketitle

\begin{abstract}
We consider several problems related to packing forests in graphs.
The first one is to find $k$ edge-disjoint forests of maximum size in a directed graph $G$ such that the indegree of each vertex in these forests is at most $k$. We describe a min-max characterization for this problem and show that it can be solved in almost linear time for fixed $k$, extending the algorithm of [Gabow, 1995]. Specifically, the complexity is $O(k \delta m \log n)$, where $n, m$ are the numbers of vertices and edges in $G$, respectively, and $\delta = \max\{1, k - k_G\}$, where $k_G$ is the edge connectivity of the graph.
Using our solution to this problem, we obtain algorithms for two applications:

(1) The {\sc $k$-forest} problem: find $k$ forests in an {\em undirected} graph $G$ maximizing the number of edges in their union. We show how to solve this problem in $O(k^3 \min\{kn, m\} \log^2 n + k \cdot{\rm MAXFLOW}(m, m) \log n)$ time, breaking the $O_k(n^{3/2})$ complexity barrier of previously known approaches.

(2) The {\sc Directed edge-connectivity augmentation} problem: find a smallest set of directed edges whose addition to the given directed graph makes it strongly $k$-connected. We give a deterministic algorithm with complexity $O(k \delta m \log n)+T_{\rm split}$, where $T_{\rm split}$ is the time to split off one vertex in a graph with $O(m+\delta n)$ edges while preserving strong connectivity among the original vertices. For Eulerian directed graphs, the total complexity is $O(k \delta m \log n)$. A similar approach solves the undirected version of the problem in $O(k \delta m \log n)$ time.
\end{abstract}

\tableofcontents

\section{Introduction}\label{sec:intro}

In this paper, we will work with an unweighted graph $G=(V,E)$ that contains $n$ vertices and $m$ edges, with parallel edges allowed.
We will consider both directed and undirected graphs; the type of $G$ will be specified when needed.
Let $k=O({\tt poly}(n))$ be a given positive integer.
Let $k_G$ be the edge connectivity of $G$, i.e.,\ the maximum integer such that $G$ is $k_G$-connected. If $G$ is directed, we will say ``$k_G$-connected'' meaning ``strongly $k_G$-connected''. We define $\delta = \max\{1, k - k_G\}$.
For a directed graph $G=(V,E)$ and a subset $S\subseteq V$, let $\rho(S)$ be the set of edges in $G$
going from $V-S$ to $S$, and for a subset $F\subseteq E$ let $\indeg_F(v)$ be the indegree of node $v$ in the graph $(V,F)$.

This paper studies algorithms for packing forests in $G$,
i.e.,\ computing edge-disjoint forests $F_1, \ldots, F_k \subseteq E$ in $G$ that maximize the size of their union $F = F_1 \sqcup \ldots \sqcup F_k$, subject to certain constraints. We treat forests as their edge sets. If $G$ is directed, then $F_i\subseteq E$ is a forest if it is acyclic in the undirected sense.
Packing forests in $G$ is intimately related to connectivity problems on $G$. Fix a node $a\in V$ in a directed graph $G$. An {\em $a$-cut} is a non-empty subset $A\subseteq V-\{a\}$.
By the classical result of Edmonds~\cite{Edmonds:69,Edmonds:72},
$|\rho(A)|\ge k$ for all $a$-cuts if and only if
$G$ has a {\em complete $k$-intersection for $a$}, i.e.,\
disjoint \emph{spanning} trees $F=F_1\sqcup \ldots\sqcup F_k$
such that $\indeg_F(a)=0$ and $\indeg_F(v)=k$ for all $v\in V-\{a\}$\footnote{The version of the theorem with $k$ edge-disjoint arborescences rooted at $a$ is perhaps more well-known.}.
If $G$ is undirected, then one could convert it to a directed graph $\vec G$
by replacing each edge of $G$ with two directed edges in opposite directions,
and then apply Edmonds' result to $\vec G$; clearly, the connectivity of $G$ equals
the cost of a minimum $a$-cut in $\vec G$ for any $a\in V$.

An efficient algorithm for computing a complete $k$-intersection $F_1, \ldots, F_k$
has been given by Gabow~\cite{gabow}. Its complexity is $O(km\log n)$
for general directed graphs, and  $O(km\log(n^2 / m))$ for directed graphs without parallel edges. 
Gabow iteratively grows the forests until they become spanning trees.
If $G$ does not have a complete $k$-intersection, then Gabow's algorithm terminates
and outputs an $a$-cut $A$ with $|\rho(A)|\le k-1$.
It is natural to ask what happens if one does not terminate Gabow's algorithm upon discovering a cut: (i) whether one can extract more information about $G$ and the collection of forests $F_1, \ldots, F_k$, and (ii) whether these forests can be used for some applications.
These questions are the main topic of this paper.
Let us consider the following problem, which will later be used as a subroutine in two applications.

\begin{myproblem}{\customlabel{the_one_problem}{{\sc A}}}{\textsc{Bounded Indegree }$k$\textsc{-Forest}}
    Given an input $(G,k,\tau)$, where $G$ is a directed graph and $\tau$ is a vector $V\rightarrow\{0,1,\ldots,k\}$,
    find $k$ edge-disjoint forests $F_1, \ldots, F_k$ such that $\indeg_F(v) \leq k-\tau(v)$ for each $v\in V$,
    and $|F| = |F_1\sqcup \ldots\sqcup F_k|$ is maximized.
\end{myproblem}
For a node $a\in V$, let $\tau^{a:k}$ be the vector with $\tau^{a:k}(a)=k$ and $\tau^{a:k}(v)=0$ for $v\in V-\{a\}$.
In these terms, Edmonds' result can be reformulated as follows: the cost of a minimum $a$-cut in $G$ is at least $k$ if and only if an optimal solution $F$
of Problem~\ref{the_one_problem} for  $(G,k,\tau^{a:k})$
satisfies $|F|=nk-k$.

We will describe an efficient algorithm for solving Problem~\ref{the_one_problem}
for the cases $\tau={\bf 0}$ and $\tau = \tau^{a:k}$, extending an algorithm of Gabow~\cite{gabow}
%Regarding Probem ~\ref{the_one_problem}, our main contribution is the complexity analysis of the algorithm, which turned out to be significantly more tedious than the one in \cite{gabow}. 
\footnote{
Our algorithm can be generalized to work for arbitrary vectors $\tau:V\rightarrow\{0,1,\ldots,k\}$ with the same complexity, but the analysis would require some extra work.
For the sake of clarity and brevity, we focus on the special case $\tau={\bf 0}$, which is sufficient for the applications that we consider.}.
\begin{restatable}{theorem}{boundedIndegreeAlgorithm}\label{th:the_one_problem:alg}
There exists a deterministic algorithm for Problem~\ref{the_one_problem} with $\tau={\bf 0}$ or with $\tau=\tau^{a:k}$ that runs in $O(k \delta m \log n)$ time.
%where $\delta=k-k_G\ge 1$ and $k_G$ is the maximum integer in $\{0,1,\ldots,k-1\}$ such that $G$ is strongly $k_G$-connected.
% Consequently, the problem can be solved in the same time for $\tau=\tau^{a:k}$ by applying the algorithm for $\tau={\bf 0}$ to the graph obtained from $G$ by removing all edges pointing to~$a$.
\end{restatable}
When using the algorithm from the theorem above in other applications, we will need a description of the structure of an optimal solution. Thus, we establish the following min-max characterization for arbitrary $\tau$.
A {\em subpartition of $V$} is a collection $\calX=\{X_1,\ldots,X_t\}$ 
of nonempty subsets of $V$ that are pairwise disjoint. For a subset $A\subseteq V$, we denote $\tau(A)=\sum_{v\in A}\tau(v)$.
\begin{restatable}{theorem}{boundedIndegreeMinMax}\label{th:minmax:kforests}
If $F$ is an optimal solution of Problem \ref{the_one_problem}, then
\begin{equation}\label{eq:minmax:kforests:intro}
nk-\tau(V)-|F|=\max_{\raisebox{-5pt}{\mbox{\small $\calX$: a subpartition of $V$}}} \;\; \sum_{A\in\calX}(k-\tau(A)-|\rho(A)|)
\end{equation}
\end{restatable}
Note that applying this theorem to the vector $\tau=\tau^{a:k}$ yields Edmonds' result stated earlier.
Indeed, the LHS of eq.~\eqref{eq:minmax:kforests:intro} is zero if and only if $G$ has a complete $k$-intersection for $a$,
while the RHS is zero if and only if $|\rho(A)|\ge k$ for all $a$-cuts $A$.

In the next two subsections we give two applications of the algorithm from Theorem~\ref{th:the_one_problem:alg}.

\subsection{\texorpdfstring{$k$}{k}-{\sc forest} problem}

\begin{myproblem}{\customlabel{problem_k_forest}{{\sc B}}}{$k$-{\sc forest}}
    Given an input $(G,k)$, where the graph $G$ is undirected, find $k$ edge-disjoint forests $F_1, \ldots, F_k$ in $G$ such that $|F| = |F_1\sqcup \ldots\sqcup F_k|$ is maximized.
\end{myproblem}
Without loss of generality, we will assume that $G$ is connected (and thus, $n$ is $O(m)$). Otherwise, it is enough to solve the problem in each connected component of $G$. 

A related problem is $k$-{\sc spanning tree}:
find $k$ edge-disjoint spanning trees in $G$, if they exist. Clearly, any algorithm for $k$-{\sc forest}
can be used to solve $k$-{\sc spanning tree}.
These problems have been extensively studied. Both are special cases of matroid union;
for example, $k$-{\sc forest} asks to find a maximum independent set
in the matroid which is a union of $k$ graphic matroids.
Thus, it can be solved in polynomial time with matroid theory techniques.
For general matroids on $m$ elements, the matroid union problem can be solved using
$\widetilde{O}(km \sqrt{p})$ independence oracle
queries, where $p$ is the size of the optimal solution~\cite{terao}. For graphic matroids, one has $p=O(kn)$.
We refer to~\cite{terao} for algorithms with other trade-offs, and for a good literature review.

Table~\ref{table:lit} lists papers that specialized matroid union algorithms to graphic matroids.
The current fastest algorithms include those of Gabow and Westermann~\cite{GabowWestermann:92}, which offer various trade-offs
for parameters $(k,n,m)$, 
and those of Blikstad, Mukhopadhyay, Nanongkai, and Tu~\cite{blikstad} and Quanrud~\cite{quanrud},
which both run in $\widetilde O(m+(kn)^{3/2})$ time.
All these complexities include a factor of $O_k(n^{3/2})$,
which seems to be a barrier for current approaches.
In particular, the complexity jumps from $O(m)$ for $k=1$
(in which case we need to find a spanning tree)
to $\widetilde O(m+n^{3/2})$ for $k=2$ (when we need to find two disjoint spanning trees, in the case of the $k$-{\sc spanning tree} problem).

When speaking of tree packings and forest packings in undirected graphs, it is appropriate to mention the following classical structural results. The tree-packing theorem by Tutte \cite{Tutte} shows that a graph has $k$ edge-disjoint spanning trees if and only if every partition $\mathcal{P}$ of the vertex set has at least $k(|\mathcal P|-1)$ crossing edges. In a related result, Nash-Williams proved that a graph can be decomposed into $k$ edge-disjoint forests if and only if $k \geq \frac{|E(G[S])|}{|S| - 1}$ for any vertex subset $S$ with at least two vertices.

In this paper we establish the following result.
\begin{restatable}{theorem}{undirectedForestsAlgorithm}\label{th:main-analysis}
\begin{sloppypar}
There exists a deterministic algorithm for
Problem~\ref{problem_k_forest} that runs in $O(k^3\min\{kn,m\}\log^2n + k\cdot{\rm MAXFLOW}(m, m) \log n)$ time,
%\footnote{Each max-flow instance has $O(m)$ vertices and $O(m)$ edges, which explains the two arguments $m,m$.}
where ${\rm MAXFLOW}(n, m)$ is the complexity of a max flow computation in a graph with $n$ vertices and $m$ edges.
\end{sloppypar}
\end{restatable}
In the algorithm, we make $O(k \log n)$ calls to a max flow algorithm
on a special graph with $O(m)$ vertices and $O(m)$ edges.
With the recent almost linear time max flow algorithms \cite{almost_linear_maxflow, deterministic-max-flow}, the complexity becomes $O(k^3\min\{kn,m\}\log^2n + k m^{1+o(1)})$.
Thus, we break the $O_k(n^{3/2})$ barrier for fixed $k$ and sufficiently sparse graphs (with $m=o(n^{3/2-o(1)})$).

\begin{table}[!htbp]
\center
\small
\begin{tabular}{|l|c@{\hspace{4pt}}l|}
\hline
paper &  & complexity \\ \hline 
%\hline \\
Imai~\cite{Imai:83} & $\dag$ & $O(k^2 n^2)$ \\
                 &  & $O(k^2 n m)$ \\
Roskind-Tarjan~\cite{RoskindTarjan:85} & $\dag$ & $O(k^2n^2)$ \\
Gabow-Stallmann~\cite{GabowStallman:85} & $\dag$ &  $O(k^{3/2} n^{1/2} m)$ if $m=\Omega(n\log n)$ \\
Gabow-Westermann~\cite{GabowWestermann:92} &  & $\widetilde O(\min\{
	\; k^{3/2} n m^{1/2} \; ,
	\; k^{1/2}  m^{3/2} \; ,
	\; k n^2 \; ,
	\; m^2 / k \; 
	\})$ \\
Blikstad-Mukhopadhyay-Nanongkai-Tu~\cite{blikstad} &  & $\widetilde{O}(m + (kn)^{3/2})$ \\
Quanrud~\cite{quanrud} &  & $\widetilde{O}(m + (kn)^{3/2})$ \\
{\bf this work} &  & $O(k^3\min\{kn,m\}\log^2n + k m^{1+o(1)})$ \\
\hline
\end{tabular}
\caption{Algorithms for $k$-{\sc forest} and $k$-{\sc spanning tree} problems. Those that handle only the latter problem are marked with ``\,$\dag$\,''.
The exact complexity of the algorithm in~\cite{GabowWestermann:92} is  
$O(\min\{ \; k^{3/2} n' \sqrt{m'} \; , \; k (n')^2 \log k\;\})$ where $n' = \min\{n, 2m / k\}$, $m' = m + n' \log n'$.
}\label{table:lit}
\end{table}

\paragraph{Applications}
The $k$-{\sc forest} problem has many applications; the following list comes from~\cite{GabowWestermann:92}.
In the analysis of
electrical networks, the solution for $k = 2$ (called {\em extremal} trees \cite{KK}) is central
to hybrid analysis (it gives the minimum fundamental equations of the network
\cite{OIW,IF}). The $k$-{\sc forest} problem also arises in the study of rigidity of structures.
For $k = 2$, it determines if a graph is rigid as a bar-and-joint framework on the
surface of a cylinder in three-dimensional space $\mathbb R^3$ \cite{Wh}. For $k = d(d + 1)/2$, it
determines if a graph is rigid as a bar-and-body framework \cite{Tay} and also as a
body-and-hinge framework \cite{Wh} in $\mathbb R^d$. For arbitrary $k$, it determines rigidity of
bar-and-joint frameworks on the flat torus in $\mathbb R^k$ \cite{Wh} and also rigidity of $k$-frames
\cite{WW}.

We also point out that packing spanning trees can be applied to routing problems. For a $k$-connected graph, one can have a $(\lfloor k / 2 \rfloor - 1)$-resilient routing table \cite{routing} based on packing arborescences in $\Vec{G}$. However, there exist $\lfloor k / 2 \rfloor$ disjoint spanning trees in $G$ by Nash-Williams' theorem, which can be used for routing to achieve the same resiliency.

\subsection{Edge connectivity augmentation problem}
Our second application does not explicitly involve packing trees, but turns out to be closely related to Problem~\ref{the_one_problem}.

\begin{myproblem}{\customlabel{problem_augmentation}{{\sc C}}}{\sc Directed Edge Connectivity Augmentation}
    Given a directed graph $G$ and an integer $k$, find a smallest set of directed edges whose addition to $G$ makes $G$ strongly $k$-connected.
\end{myproblem}

\begin{myproblem}{\customlabel{undirected_problem_augmentation}{{\sc C$'$}}}{\sc Undirected Edge Connectivity Augmentation}
    Given an undirected graph $G$ and an integer $k$, find a smallest set of undirected edges whose addition to $G$ makes $G$  $k$-connected.
\end{myproblem}
When stating complexities for these problems, we will assume that $m=\Omega(n)$
(to simplify expressions); if this is not the case, then $m$ would need to be replaced by $m+n$ in the runtimes.

A polynomial-time algorithm for Problem~\ref{problem_augmentation}
was first given by Frank~\cite{frank}, who also provided a min-max characterization.
The complexity of Frank's algorithm has been improved by Gabow in~\cite{Gabow:STOC91} and in \cite{Gabow:STOC94}.
%Let $k_G$ be the maximum integer in $\{0,1,\ldots,k-1\}$ such that $G$ is strongly $k_G$-connected,
%and let $\delta=k-k_G\ge 1$ be the desired increase in the connectivity. 
%To simplify expressions for complexities, we assume that $m=\Omega(n)$; if it is not the case, then $m$ would need to be replaced by $m+n$ in the runtimes.
The latter algorithm runs in $O(k \delta (m+\delta n)\log n)$ time. Gabow~\cite{Gabow:STOC94} also presented a strongly polynomial algorithm
with complexity $O(n^2 m\log (n^2/m))$  
(assuming that parallel edges are represented by a single edge with a weight, i.e.,\ for directed weighted graphs).

Polynomial-time algorithms for Problem~\ref{undirected_problem_augmentation}  have been given in~\cite{watanabe-nakamura,CaiSun:89,NaorGusfieldMartel:89,Gabow:STOC91,Gabow:STOC94,NagamochiIbaraki:97,BenczurKarger:00,monte_carlo_augmentation}.
The best known deterministic algorithms have complexity $O(k \delta (m+\delta n)\log n)$ \cite{Gabow:STOC94}
and $\tilde O(mn)$ \cite{NagamochiIbaraki:97}, while the fastest randomized (Monte-Carlo) algorithm runs in $\tilde O(m)$~\cite{monte_carlo_augmentation}.

We will show that Problem~\ref{problem_augmentation} is closely related to Problem~\ref{the_one_problem} with the vector $\tau={\bf 0}$.
Using this connection, we will establish the following result.
\begin{theorem}
There exists a deterministic algorithm for Problem~\ref{problem_augmentation} with complexity $O(k \delta m \log n)+T_{\rm split}$, where $T_{\rm split}$ is the time to split off a vertex $s$ in a digraph with $O(m+\delta n)$ edges while preserving strong $k$-connectivity among the remaining vertices (see Section~\ref{sec:directed}). If $G$ is Eulerian, the total complexity is $O(k \delta m \log n)$. Problem~\ref{undirected_problem_augmentation} can also be solved deterministically in $O(k \delta m \log n)$ time.
\end{theorem}
A splitting operation at a vertex $s$ replaces two directed edges $(u,s)$ and $(s,v)$ with a single edge $(u,v)$. To split off $s$ means to apply this operation repeatedly until $s$ is isolated, then delete $s$. We give details in Section~\ref{sec:directed}.

\section{High-level overviews}

Throughout the paper we use the following notation. If $E$ is a set of directed edges, then $E^{\tt rev}$ is the set obtained
from $E$ by reversing edge orientations. Similarly, for a directed graph $G=(V,E)$ we denote $G^{\tt rev}=(V,E^{\tt rev})$.
For a subset $A\subseteq V$ we define the following sets of edges:
\begin{subequations}
\begin{eqnarray}
\rho(A)&=& E(V-A, A) = \{(u,v)\in E\::\:u\notin A,v\in A\} \label{eq:rhoA:def} \\
\rho^{\tt rev}(A)&=& E(A, V-A) = \{(u,v)\in E\::\:u\in A,v\notin A\} \label{eq:rhorevA:def}  \\
E(A)&=& E(G[A]) = \{(u,v)\in E\::\:u\in A,v\in A\} \label{eq:lambda:def} 
\end{eqnarray}
\end{subequations}
Note that definition~\eqref{eq:lambda:def} can be used for both directed and undirected graphs,
while~\eqref{eq:rhoA:def} and~\eqref{eq:rhorevA:def} are applicable only to directed graphs.
For singleton sets we will write $\rho(v)$, $\rho^{\tt rev}(v)$, $A+v$, $A-v$ instead of
$\rho(\{v\})$, $\rho^{\tt rev}(\{v\})$, $A\cup\{v\}$, $A-\{v\}$.

The main result of this paper is a description of a way to use the \textsc{Bounded Indegree $k$-Forest} problem as a subroutine to solve the \textsc{$k$-Forest} problem and the \textsc{Directed Edge Connectivity Augmentation} problem. The main insight that enables these connections is a dual min-max characterization of the optimal solution to the \textsc{Bounded Indegree $k$-Forest} problem. The main difficulty, however, is the proof that the \textsc{Bounded Indegree $k$-Forest} problem can be solved in nearly linear time in $m$. At first glance, it is very natural to expect that Gabow-type acceleration will work in this problem, with almost the same argument. However, the argument does not carry over directly, and new arguments are needed, as we discuss below.

Next, we give a high-level overview of the ideas behind our solutions to Problems~\ref{the_one_problem},~\ref{problem_k_forest} and~\ref{problem_augmentation}.
%Some proofs are moved to Appendices to meet the volume requirements.
Full descriptions will be given later in Sections~\ref{sec:directed_problem}-\ref{sec:undirected}.

\subsection{Overview of Problem~\ref{the_one_problem}: {\sc Bounded Indegree \texorpdfstring{$k$}{k}-Forest}}
Our algorithm for Problem~\ref{the_one_problem} is
very close to Gabow's algorithm from~\cite{gabow} for computing a complete $k$-intersection.
Consider the case $\tau = {\bf 0}$. As in~\cite{gabow}, we grow the solution for $k$ recursively from the solution for $k-1$. For a fixed value of $k$, the current set of forests $F_1, \ldots, F_k$ is grown by repeating the following operation:
pick a node $v$ that has a {\em deficit}, i.e.,\ its current indegree in $F$ is smaller than $k$. We will try to modify the forests to increase the indegree of this node while preserving the indegrees of other nodes. This procedure increases $|F|$.
If there exists an uncovered edge $e \in \rho(v)$ that connects two different connected components of some forest, we can just greedily add $e$ to this forest, which increases $\indeg(v)$.
But we may have a more complicated scenario.
Consider the following example: $e \in \rho(v)$ is uncovered, and adding it to the forest $F_1$ creates an (undirected) cycle in $F_1$, which can be broken if we remove some other edge $f$ on this cycle from $F_1$. Then $\indeg_F (v)$ is increased by $1$, but $\indeg_F (\head(f))$ is decreased by 1, so we have ``moved'' the deficit from $v$ to $\head(f)$.
Then, if there is an uncovered edge $g$ that enters $\head(f)$ and that can be added to some forest without creating a cycle, we can add it and decrease the deficit.
In general, we have a sequence of such movements of the deficit until we can decrease it. The sequence of edges that we use is called an \emph{augmenting path}, as in flow algorithms or matching algorithms. This example omits some important algorithmic details; the complete description is given in Section~\ref{sec:directed_problem}.

% and search for an augmenting path from $v$ until we hit another connected component of $F_k$.
% This path is then augmented, which merges two connected components of some forest $F_i$.
% The algorithm preserves the nestedness of connected components of forests $F_1,\ldots,F_k$.

The main challenge here is the analysis of this algorithm. In Gabow's case, it holds directly by construction that the forests $F_1,\ldots,F_{k-1}$ are spanning trees,
and each component of the last forest~$F_k$ has exactly one deficient vertex. This ensures
that each component of $F_k$ is processed at most once in each round, and each round halves
the number of components of $F_k$.

If $G$ does not have a complete $k$-intersection, then the forests $F_1,\ldots,F_{k-1}$
may no longer be spanning trees, and so the structural property above does not hold anymore. Now augmentations may increase not only the last forest $F_k$, but also other forests.
Without a bound on the number of deficits in each component, the complexity may become much worse.
This difficulty might have been the reason why this rather natural generalization of Gabow's algorithm
has not appeared in the literature before.

We overcome this difficulty by establishing a structural result that gives a guarantee of at most~$k$
deficient vertices per  component of $F_k$. This leads to the $O(k \delta m \log n)$ runtime  
claimed in Theorem~\ref{th:the_one_problem:alg}.

\subsection{Overview of Problem~\ref{problem_k_forest}: {\sc \texorpdfstring{$k$}{k}-forest}}

For this problem, $G$ is undirected. Notice that if one introduces any orientation on $G$, then the optimal value for Problem \ref{the_one_problem} (with $\tau = {\bf 0}$) on this oriented graph is not greater than the optimal value of Problem \ref{problem_k_forest}, because Problem \ref{the_one_problem} has additional indegree constraints. 

On the other hand, suppose that an optimal solution $F_1, \ldots, F_k$ to Problem \ref{problem_k_forest} is known. If we orient $G$ such that each component of each $F_i$ becomes oriented away from some root chosen in this component, then the optimal values for Problem \ref{problem_k_forest} and Problem \ref{the_one_problem} will coincide, since $F_1, \ldots, F_k$ would be a feasible solution to Problem \ref{the_one_problem}.

So, under a suitable orientation of a part of $G$, one can get an optimal solution of Problem \ref{problem_k_forest} as an optimal solution of Problem \ref{the_one_problem}. We build an iterative process that finds such an orientation. Each iteration consists of the following steps:
\begin{enumerate}
    \item Find a maximal subgraph $P \subseteq G$ that contains the current solution $F$ and admits an orientation satisfying $\indeg_P (v) \leq k$ for all $v \in V$. This can be done in one max flow computation on a graph with $O(m)$ edges.

    \item Solve Problem \ref{the_one_problem} in this subgraph. This updates the current set of forests.

    \item Contract all sets $U \subseteq V$ such that each $F_i$ is spanning inside $U$. This ensures that progress will be made if the current solution is not optimal.
\end{enumerate}

We analyze this process and show that it converges in $O(k \log n)$ iterations, yielding the complexity bound $O(k^3 \min\{kn, m\} \log^2 n + k \cdot{\rm MAXFLOW}(m, m) \log n)$.

\subsection{Overview of Problem~\ref{problem_augmentation}:  {\sc Directed Edge Connectivity Augmentation}} 
The idea of the algorithm is based on the following min-max characterization by Frank~\cite{frank}.
Below, a subpartition $\calX$ of $V$ is called {\em proper} if $\calX\ne\{V\}$.
\begin{theorem}[{\cite[Theorem 3.1]{frank}}]
Let  $\gamma (G, k)$ be the optimal value of Problem~\ref{problem_augmentation},
i.e.,\ the minimum number of directed edges that should be added to $G$ to make it strongly $k$-connected. Then
    \begin{equation}
        \gamma(G, k) = \max \left\{ \alpha_{\tt in}(G, k), \alpha_{\tt out}(G, k) \right\},
    \end{equation}
    where 
    \begin{eqnarray}
	\alpha_{\tt in}(G, k)&=&\max_{\raisebox{-5pt}{\mbox{\small $\calX$: a proper subpartition of $V$}}} \;\; \sum_{A\in\calX}(k-|\rho(A)|) \label{eq:minmax:augmentation:intro} \\
	\alpha_{\tt out}(G, k)&=&\max_{\raisebox{-5pt}{\mbox{\small $\calX$: a proper subpartition of $V$}}} \;\; \sum_{A\in\calX}(k-|\rho^{\tt rev}(A)|) 
	\;\;=\;\; \alpha_{\tt in}(G^{\tt rev}, k)
%            \alpha_{\tt in}(G, k) = \max \sum_i (k - |\rho (X_i)|), \\
%            \alpha_{\tt out}(G, k) = \max \sum_i (k - |\delta (X_i)|),
    \end{eqnarray}
    %where maximums are taken over subpartitions $\mathcal{X} = \{X_i\}_{i = 1}^t$ of $V$, $\mathcal{X} \neq \{V\}$. 
    \label{theorem_frank}
\end{theorem}
One can now see a connection between Problem~\ref{the_one_problem} with the zero vector $\tau={\bf 0}$ and Problem~\ref{problem_augmentation}:
the right-hand sides of~\eqref{eq:minmax:kforests:intro} and \eqref{eq:minmax:augmentation:intro}
look almost identical, except that in~\eqref{eq:minmax:kforests:intro} the maximization is over
all subpartitions $\calX$, while in~\eqref{eq:minmax:augmentation:intro} it is over {\bf proper}
subpartitions.\footnote{Condition $\mathcal{X} \neq \{V\}$ is missing
    in the formulation of~\cite[Theorem 3.1]{frank}, but is explicitly stated in \cite[Theorem 63.1]{schrijver-book:B}.}
This means, in particular, the following: if we compute an optimal solution $F$ for $(G,k,{\bf 0})$
and this solution happens to contain a disconnected forest (equivalently, if $|F|<nk-k$)
then $\alpha_{\tt in}(G,k)=nk-|F|>k$, since the optimal subpartition in~\eqref{eq:minmax:kforests:intro} cannot be $\{V\}$.
 Otherwise we have $\alpha_{\tt in}(G,k)\le k$,
and some other method is needed.

In the general case, we solve Problem~\ref{the_one_problem} for $(G,k,\tau^{a:k})$
and then postprocess the solution.
We no longer claim that this computes $\alpha_{\tt in}(G,k)$;
instead, we claim to correctly compute $\gamma(G,k)=\max \left\{ \alpha_{\tt in}(G, k), \alpha_{\tt out}(G, k) \right\}$.

Once $\gamma(G,k)$ is computed, we need to compute the actual set of edges to be added.
As in previous work, we do it by calling an {\em edge splitting} procedure
in a specially constructed  graph; details will be given in Section~\ref{sec:directed}.

%In the next sections we give detailed descriptions of our solutions to Problems~\ref{the_one_problem},~\ref{problem_k_forest},~\ref{problem_augmentation} and~\ref{undirected_problem_augmentation}.
%Proofs that are missing in these sections can be found in Appendices~\ref{appendix_directed_problem} and~\ref{appendix_k_forests}.

%We say that a collection of sets $\mathcal{X} = (X_1, \ldots, X_t)$ is a \textit{subpartition} of a set $V$ if $X_i \subseteq V$ for all $i$, and $X_i \cap X_j = \varnothing$ for all $i \neq j$.
%This subpartition is called \textit{proper} if $\calX\ne\{V\}$.

%We will need this $\tau$ later. Notice that the case $\tau = \tau^{a:k}$ can be effortlessly reduced to the case $\tau = {\bf 0}$.
%
%\begin{observation}
%    Problem \ref{the_one_problem} for $(G, k, \tau^{a:k})$ can be solved as follows. Delete all edges pointing to $a$, then solve the problem for $(G, k, {\bf 0})$.
%    \label{obs_tau_reduction}
%\end{observation}

%%%%%%%%%%%%%%%%%%%%%%%%%%%%%%%%%%%%%%%%%%%%%%%%%%%%%%%%%%%%%%%%%%%%%%%%%%%%%%%%%%%%%%%%%%%%%%%%%%%%%%%%%%%%%%

%%%%%%%%%%%%%%%%%%%%%%%%%%%%%%%%%%%%%%%%%%%%%%%%%%%%%%%%%%%%%%%%%%%%%%%%%%%%%%%%%%%%%%%%%%%%%%%%%%%%%%%%%%%%%%

%Frank showed in~\cite{frank} that Problem 
%The undirected case can also be solved with techniques of extensions and edge splitting, and it appears to be computationally easier. There is a very fast $\widetilde{O}(m)$ %randomized Monte-Carlo algorithm that solves the undirected augmentation problem \cite{monte_carlo_augmentation}.

\section{Algorithm for Problem~\ref{the_one_problem}}\label{sec:directed_problem}
In this section we discuss an algorithm that, given a directed graph $G=(V,E)$, computes a set of edges $F=F_1\sqcup\ldots\sqcup F_k\subseteq E$
of maximum cardinality such that $F_1,\ldots,F_k$ are edge-disjoint forests of $G$,
and $\indeg_F(v)\le k-\tau(v)$ for each node $v\in V$, where
$\tau$ is a fixed vector $V\rightarrow\{0,1,\ldots,k\}$. %We will give the detailed algorithm and its analysis only for the case $\tau = {\bf 0}$. 
Whenever we write $F=F_1\sqcup\ldots\sqcup F_k$, we assume, with some abuse of notation, that the integer $k$ and the forests $F_1,\ldots,F_k$ can be recovered from~$F$.

The problem can be stated in matroid terms as follows.
Let $\calG$ be the graphic matroid corresponding to $G$,
and let $\calG^k$ be the $k$-fold union of $\calG$.
The independent sets of the matroid $\calG^k$ correspond to the subgraphs of $G$ that decompose into $k$ forests.
Let $\calD$ be the matroid on the edge set of $G$, where $H\subseteq E$ is independent if $\indeg_H(v) \leq k-\tau(v)$ for each vertex $v$. 
Problem \ref{the_one_problem} then asks to find a set $F\in\calG^k\cap\calD$
of maximum cardinality.
% maximal set independent in both $\mathcal{G}^k$ and $\mathcal{D}^k$, which is known as the matroid intersection problem. 
Clearly, it can be solved in polynomial time
by applying an algorithm for the matroid intersection problem.
It requires oracles that test, for a given $F\subseteq E$, whether $F\in\calG^k$ and $F\in\calD$;
the first one can be implemented using a matroid union algorithm.

To obtain the desired running time, we use a more direct approach.
In Section~\ref{sec:kforests:augmentations} we describe a generic scheme for solving the problem,
which is based on performing augmentations in a certain auxiliary graph.
In Section~\ref{sec:min-max} we prove some useful properties of an optimal solution.
Then in Section~\ref{sec:kforests:alg} we present an efficient way of organizing these augmentations (for $\tau={\bf 0}$).

We will use the following terminology and notation. 
%For a subset $A\subseteq V$ we denote
%\begin{eqnarray*}
%\gamma(A) &=& \{e \in E \; | \; \head(e) \in A, \tail(e) \in A\} \\
%\rho(A) &=& \{e \in E \; | \; \head(e) \in A, \tail(e) \notin A\}
%\end{eqnarray*}
The edges of $E - F$ will be called \textit{uncovered}. If an edge $e \notin F_i$ can be added to $F_i$ without creating any undirected cycles
(i.e.,\ its endpoints lie in different components of $F_i$), we will call $e$ \textit{joining for} $F_i$.
The \textit{deficit} of a vertex $v$ in $F=F_1\sqcup \ldots \sqcup F_k$ is defined as
    \begin{equation}
        \Def_F(v) = k - \tau(v) - \indeg_F(v) %\in \{0,1,\ldots,k\}
    \end{equation}
Vertices with $\Def_F(v)>0$ are called \textit{deficient}.
For a subset $A\subseteq V$, we also denote $\Def_F(A)=\sum_{v\in A}\Def_F(v)$.
%If $\tau$ is not clear from the context, then we may write $\Def^\tau_F(A)$ instead of $\Def_F(A)$.

\subsection{Auxiliary graph and augmentations}\label{sec:kforests:augmentations}

\begin{definition}[auxiliary graph]
    For given forests $F_1, \ldots, F_k$ with $F=F_1\sqcup\ldots\sqcup F_k$, the auxiliary graph $D(F)$ is a directed graph with vertex set $V\sqcup E \sqcup \{t\}$ and the following edges (below $v\in V$, $e,f\in E$):
    \begin{enumerate}
        \item $(v, e)$ is an edge if $\indeg_F(v) < k-\tau(v)$, $\head(e)=v$ and $e \notin F$.
        \item $(e, t)$ is an edge if $e$ is joining for some $F_i$.
        \item $(e, f)$ with $f\notin F$ is an edge if $e \in F$ and $\head(e) = \head(f)$.
        \item $(f, e)$ with $e\in F$ is an edge if $e \in F_i$, $f \notin F_i$ and $F_i - e + f \in \mathcal{G}$ for some $F_i$.
    \end{enumerate}
\end{definition}

Let $P = (v, e_1, \ldots, e_r, t)$ be a $V$-$t$ path in the auxiliary graph. We will need the following \textit{augmentation} procedure\footnote{This ``augmentation" has nothing to do with ``augmentations" in the Augmentation Problem \ref{problem_augmentation}, but both are conventional terminology. We hope this will not cause confusion.}.
We assume below that the solution $F=F_1\sqcup\ldots\sqcup F_k$
is represented by labels $\ell(e)\in [k]\cup\{\varnothing\}$ for edges $e\in E$,
where $\ell(e)=i$ if $e\in F_i$, and $\ell(e)=\varnothing$ if $e\notin F$.

\begin{algorithm}[!htbp]
    \DontPrintSemicolon
    \SetNoFillComment

    $p \gets \min\{j \; | \; e_r \text{ is joining for } F_j \}$ \\
    let $(\ell_1,\ell_2,\ldots,\ell_r)$ be the labels of $(e_1,\ldots,e_r)$ in $F$ (then $\ell_1=\varnothing$) \\
    update $F$ by setting the labels of these edges to $(\ell_2,\ldots,\ell_r,p)$
%    \For{$i \in \{1, \ldots, r - 1\}$}
%    {
%        {\bf if} {$e_{i + 1} \in F_j$} {\bf then} {assign $e_i$ to $F_j$} \\
%        {\bf if} {$e_{i + 1}$ is uncovered} {\bf then} {make $e_i$ uncovered} \\
%%        \If{$e_{i + 1} \in F_j$}{assign $e_i$ to $F_j$}
%%        \If{$e_{i + 1}$ is uncovered}{make $e_i$ uncovered}
%    }
%    assign $e_r$ to $F_p$.
    
    \caption{{\tt Augment}($P$: $(v, e_1, \ldots, e_r, t)$ -- a $V$-$t$ path).}
    \label{proc_augment}
\end{algorithm}

The following theorem can be shown using standard techniques from matroid theory (see Appendix~\ref{sec:proc_augment:proof}).
\begin{theorem}\label{th:proc_augment}
    {\rm (a)} Algorithm~\ref{proc_augment} 
    does not change the sizes of forests $F_i$ for $i \neq p$, and increases the size of $F_p$ by one.
    Furthermore, it increases $\indeg_F(v)$ by 1, and does not change the indegrees of other vertices. \\
{\rm (b)} If $P$ is a $V$-$t$ path in the auxiliary graph without shortcuts, then augmenting $P$ produces feasible forests $F'=F'_1\sqcup\ldots\sqcup F'_k$
(with $|F'|=|F|+1$). \\
{\rm (c)}   If the auxiliary graph does not have $V$-$t$ paths, then $F$ is optimal.
\end{theorem}

Theorem~\ref{th:proc_augment} implies that Problem \ref{the_one_problem} can be solved % in polynomial time 
by repeatedly finding and augmenting paths. As we will discuss later, Algorithm~\ref{proc_augment} runs in $O(r + k \alpha(n))$ amortized time, where $\alpha$ is the inverse Ackermann function. The computation of $p \gets \min\{j \; | \; e_r \text{ is joining for } F_j \}$ in the first line can be done using $k$ disjoint-set data structures, one structure per forest. The updates of $r$ edge labels take $O(r)$ time.
%In section \ref{sec:kforests:alg} we will show how to do it efficiently for the case $\tau = {\bf 0}$.

\subsection{Min-max characterization}\label{sec:min-max}
In this section, we establish some properties of optimal solutions of Problem~\ref{the_one_problem}.

\begin{definition}\label{def:closed}
Consider forests $F=F_1\sqcup \ldots\sqcup F_k$. A subset $A\subseteq V$ is {\em $F$-closed} if $\rho(A)\subseteq F$
and each $F_i$ is a spanning tree in $A$ (i.e.,\ the subgraph $(A,F_i\cap E(A))$ is connected).
\end{definition}
Such a set $A$ satisfies $\indeg_F(A)=|\rho(A)|+k(|A|-1)$. This yields the following equation:
\begin{equation}\label{eq:Fclosed}
\Def_F(A)=k-\tau(A)-|\rho(A)| \qquad\quad \forall A\subseteq V,\mbox{ $A$ is $F$-closed}
\end{equation}

\begin{lemma} \label{lem_q}
{\rm (a)} Let $F=F_1\sqcup \ldots\sqcup F_k$ be a feasible solution of Problem \ref{the_one_problem},
and suppose that there is no $v$-$t$ path in the auxiliary graph for $F$ for some node $v$ with $\Def_F(v)>0$.
Then there exists an $F$-closed set $Q_v\ni v$ which is contained in some component of the forest $F_i$ for every $i\in[k]$.

\noindent {\rm (b)} Let $F=F_1\sqcup \ldots\sqcup F_k$ be an optimal solution of Problem \ref{the_one_problem}. 
    There exists a subpartition $\calS$ of $V$ with the following properties:
    (i) Every set $A\in\calS$ is $F$-closed, satisfies $\Def_F(A)>0$,
    and is contained in some component of the forest $F_i$ for every $i\in[k]$.
    (ii) Every vertex $v\in V$ with $\Def_F(v)>0$ belongs to some $A\in\calS$.
%    There exist pairwise-disjoint $F$-closed sets $S_1,\ldots,S_q\subseteq V$
%    such that (i) each $S_i$ is contained in some component of $F_j$ for every $j\in[k]$, and (ii) every vertex $v\in V$ with $\Def_F(v) > 0$ belongs to $S_1\sqcup\ldots\sqcup S_q$.
\end{lemma}
\begin{proof}
    Note that part {\rm (b)} easily follows from part {\rm (a)}.
    Indeed, the desired subpartition $\calS$ can be obtained algorithmically as follows:
    (i) start with the family of sets $\{Q_v\::\:\Def_F(v) > 0\}$; (ii) while there are overlapping sets $A,B$ in the family,
    replace them with $A\cup B$. Clearly, if $A$ and $B$ are $F$-closed, then so is their union.
    
    We thus focus on proving {\rm (a)}.
    %Consider $v\in V$ with $\Def_F(v) > 0$.
    Let $L\subseteq E$ be the set of edges of $G$ (corresponding to nodes in the auxiliary graph) that are reachable from $v$ in the auxiliary graph for $F$.
    Note that for every $i\in[k]$ we have $L\subseteq E(C)$ where $C$ is the component of $F_i$ containing $v$; otherwise, the auxiliary graph contains a path $(v,\ldots,e,t)$, where the edge $e$ is joining for $F_i$.
    Let $Q_v$ be the set of vertices of $L$ (or just $\{v\}$ if $L$ is empty, but this is a trivial case).
    To prove the lemma, it suffices to show that $Q_v$ is $F$-closed.
    
    The proof of this fact will closely follow the argument from  \cite[Section 2.1]{gabow}.
First, let us show that each $F_i\cap L$ is spanning in $Q_v$. Assume the contrary. Let $Q_v = A \sqcup B$, where $A$ and $B$ are not connected in $F_i \cap L$. 
The structure of the auxiliary graph implies that  $L$ is connected. Thus, there exists an edge $e \in (L - F_i)$ connecting $A$ and $B$. We know that $e$ is not joining for $F_i$
(otherwise {\tt Search} would find an augmenting path, and $F$ would not be optimal).
 Therefore, there is a path in $F_i$ between $\head(e)$ and $\tail(e)$. Every edge $f$ of this path belongs to $L$ (since $(e,f)$ is in the auxiliary graph by construction). Then $A$ and $B$ are connected in $F_i \cap L$, which is a contradiction.

It remains to show that $\rho(Q_v)\subseteq F$. 
Suppose not, i.e.,\ there exists $e\in \rho(Q_v)- F$.
The condition $\tail(e)\notin Q_v$ implies that $e\notin L$.
Denote $u=\head(e)$. If $u=v$, then $e\in L$ by the definition of the auxiliary graph - a contradiction.
If there exists $f\in F_i\cap L$ with $\head(f)=u$, then $(f,e)$ is in the auxiliary graph and hence $e\in L$ - a contradiction.
This shows that $u\ne v$ and $\indeg_{F\cap L}(u)=0$.
We can write $\indeg_{F\cap L}(Q_v)=\indeg_{F\cap L}(Q_v-\{u,v\})+\indeg_{F\cap L}(u)+\indeg_{F\cap L}(v)\le k(|Q_v|-2)+0+(k-1)=k(|Q_v|-1)-1$.
On the other hand, each $F_i\cap L$ is spanning in $Q_v$ and hence $\indeg_{F\cap L}(Q_v)\ge k(|Q_v|-1)$.
We obtained a contradiction.
\end{proof}

Now Theorem~\ref{th:minmax:kforests} follows as a corollary of Lemma~\ref{lem_q}.
\boundedIndegreeMinMax*
We also show that the maximum in~\eqref{eq:minmax:kforests:intro} is achieved by the subpartition $\calX=\calS$ constructed in Lemma~\ref{lem_q}.
\begin{proof}
Recall that $\Def_F(V)=nk-\tau(V)-|F|$.
Consider the subpartition $\calS$ of $V$ as in Lemma \ref{lem_q}.
Summing eq.~\eqref{eq:Fclosed} over $A\in\calS$ and using the fact that the sets $A\in\calS$ cover all deficient vertices, we get
    \begin{equation}
        \Def_F(V) = \sum_{A\in\calS} (k - \tau(A)- |\rho (A)|)
    \end{equation}
    On the other hand, consider an optimal subpartition $\mathcal{X}$ of $V$ that maximizes~\eqref{eq:minmax:kforests:intro}. For each $A\in\calX$
    we have
    $
        \indeg_F (A) \leq k (|A| - 1) + |\rho (A)|
    $
    and therefore
    $
        \Def_F (A)\ge k - \tau(A) - |\rho (A)|
    $.
    This gives
    $$
    \Def_F(V) \ge \sum_{A\in\calX}\Def_F(A) \ge \sum_{A\in\calX} (k - \tau(A)- |\rho (A)|)
    $$
\end{proof}

\subsection{Main algorithm for Problem~\ref{the_one_problem} with \texorpdfstring{$\tau={\bf 0}$}{tau = 0} \texorpdfstring{or $\tau = \tau^{a:k}$}{or tau(a:k)}}\label{sec:kforests:alg}

% Throughout this subsection, $\tau = {\bf 0}$.

We will store a collection of forests $F=F_1\sqcup \ldots\sqcup F_k$, and increase $|F|$ by augmenting $V$-$t$ paths. Initially, $F_k$ is empty.

We also maintain the \textit{nestedness condition}: for any $i \in \{1, \ldots, k - 1\}$, $F_{i + 1} \subseteq {\tt span}(F_i)$, where ${\tt span}(F_i)$ is the set of all edges with both endpoints in the same connected component of $F_i$.
This is equivalent to the following:
for any $j > i$ and any component $C \subseteq V$ of $F_j$, there exists a component $C' \subseteq V$ of $F_i$ such that $C \subseteq C'$.
%Equivalently, for any $j > i$, $F_{j} \subset span(F_i)$. In other words, for any $j > i$, for any component $C \subset V$ of $F_j$, there exists a component $C' \subset V$ of $F_i$ such that $C \subset C'$.

%Let us denote the number of vertices in $C$ as $|C|$, and the number of edges of $G[C]$ as $e(C)$, with a slight abuse of usual notation. 

One can observe that after executing {\tt Augment}($P$), two components of $F_p$ get merged, and all the other components of each forest remain the same. 
(Here $p$ is the index chosen at line 1 of Algorithm~\ref{proc_augment}.) Moreover, {\tt Augment}($P$) preserves the nestedness condition.

Recall that an augmenting path can start only at a vertex with $\indeg_F(v)<k$, or equivalently with $\Def_F(v)>0$.
%If we fail to find a path from such $v$,
We will maintain an integer $\mdef(v)\in\{\Def_F(v),0\}$ for each $v\in V$,
where the condition $\mdef(v)=0$ indicates that either $\Def_F(v)=0$, or
searching for an augmenting path from $v$ has previously failed.
The following lemma shows that
augmenting paths from such $v$ will never appear again,
so it suffices to search for augmenting paths only from vertices $v$ with $\mdef(v)>0$.
\begin{lemma}
    \label{lem_hopeless_once_hopeless_forever}
    If there is no $v$-$t$ path in the auxiliary graph for a given vertex $v$, then, after any valid augmentation, there will still be no $v$-$t$ path.
\end{lemma}
\begin{proof}
    By Lemma \ref{lem_q}, if there are no augmenting paths starting from $v$, then there exists a set $Q \subseteq V$, $v \in Q$, such that each forest is a spanning tree inside $Q$, and $\rho(Q) \subseteq F$. Then $\indeg_F(Q)$ cannot increase. Indeed, all the edges entering $Q$ are already covered, and the number of covered edges inside $Q$ cannot be increased without creating a cycle in some forest. Since augmentations do not decrease the indegree of any vertex, the indegree of any vertex in $Q$ (including $v$) cannot increase.
\end{proof}

For a subset $A\subseteq V$, we will denote  
$\mdef(A)=\sum_{v\in A}\mdef(v)$.
We can now describe the {\tt Search} procedure (see Algorithm~\ref{alg_search}).
Note that it decreases $\mdef(v)$ by at least one, and does not change $\mdef(u)$ for nodes $u\neq v$.

\begin{algorithm}[!htbp]
    \DontPrintSemicolon
    \SetNoFillComment

    $v \gets $ any vertex in $C$ with $\mdef(v)>0$ \\
%    $L \gets \{e \in E - F \; | \; \head(e) = v\}$ \\
    search for a $v$-$t$ augmenting path $P=(v,e_1,\ldots,e_r,t)$ without shortcuts in the auxiliary graph $D(F)$ such that $e_1,\ldots,e_{r-1}\in E(C)$ and $e_r$ has exactly one endpoint in $C$ \\
    \tcc{if any $v$-$t$ augmenting path exists, then we can choose one satisfying these conditions}
    %Try to find an $L - t$ path among the edges of $\gamma(C)\cup\rho(C)$  \\
    \If{\text{search is successful}} {
        {\tt Augment}($P$) \\
        decrease $\mdef(v)$ by 1 to preserve the equality $\mdef(v)=\Def_F(v)$
    } \Else{
         $\mdef(v)\leftarrow 0$
    }
    
    \caption{{\tt Search}($C$: component of $F_k$ with $\mdef(C)>0$).}
    \label{alg_search}
\end{algorithm}

We implement Algorithm~\ref{alg_search} using a standard {\em cyclic scanning search} approach~\cite{GabowWestermann:92,gabow}, which takes time $O(|E(C)|)$, plus the time to compute $p=\min\{j \; | \; e_r \text{ is joining for } F_j \}$ in the augmentation procedure if the search is successful. The search stops as soon as it discovers an edge joining for $F_k$. During the search, the visited edges together with $v$ form a connected subgraph. Thus, the first edge outside $E(C)$ has exactly one endpoint in $C$ and is joining for $F_k$, so the search stops at this edge. By nestedness, no edge inside $E(C)$ is joining for any forest. Therefore, all nonterminal edges of the returned path lie in $E(C)$, and the terminal edge has exactly one endpoint in $C$.
We refer to \cite[Section 2.3-2.5]{gabow} for a detailed description of the cyclic scanning search; the fact that Gabow deals with spanning trees for $F_1, \ldots, F_{k-1}$ instead of forests does not change the validity of this approach.

Gabow's algorithm works in almost linear time in $m$ because it always maintains exactly one deficit per connected component of the last forest.
Then, one round of searching for augmenting paths reduces the number of components by half. Since all the other forests are already spanning trees, their size does not change.
After packing the next tree, the indegree of each vertex (except the root) increases by exactly one. 
In our case, however, the deficient vertices may be distributed in a much less structured way.
In general, there exist components of the last forest with zero deficit (therefore, we cannot do a search in such components at all), or with deficit greater than one. 
This is due to the fact that all the forests may increase their size, not only the last forest. A case that would lead to a high running time is a linear-size component of the last forest with a linear-size deficit. We might then have to perform many linear-time searches in this component.
The following theorem bounds the deficit of a component during our algorithm, showing that the bad case does not happen, which is crucial to the complexity analysis of our algorithm.
Its proof is somewhat lengthy, and we defer it to Section~\ref{appendix_low_deficit}.

\begin{theorem}
    Let $\widetilde{F}=\widetilde{F}_1\sqcup \ldots\sqcup \widetilde{F}_{k-1}$ be an optimal solution of Problem \ref{the_one_problem} for $k - 1$, 
    and suppose that it obeys the nestedness condition and $\Def_{\widetilde{F}}(C) \leq k - 1$ for each component $C$ of $\widetilde{F}_{k-1}$. Consider Problem \ref{the_one_problem} for $k$ with an initial configuration equal to $\widetilde{F}_1, \ldots, \widetilde{F}_{k-1}$ and the empty $k$-th forest. Then, after any sequence of augmentations performed by \textnormal{\texttt{Search}}, each component $C$ of $F_k$ satisfies $\Def_F(C) \le k$.
    \label{theorem_low_deficit}
\end{theorem}

Now we can describe the final algorithm (see Algorithm~\ref{alg_solve}).

\begin{algorithm}[!htbp]
    \DontPrintSemicolon
    \SetNoFillComment

    %$F_1, \ldots F_{k - 1} \gets$ optimal solution of Problem \ref{the_one_problem} for $k - 1$ \\
    if $k>1$ then recursively compute an optimal solution $F_1, \ldots F_{k - 1}$ of Problem \ref{the_one_problem} for $k - 1$ \\
    $F_k \gets$ empty forest, $\mdef(V) \gets \Def(V)$ \\
    \While{$\mdef(V) > 0$ and $F_k$ is not a spanning tree}{
        $\mathcal{C} \gets$ set of all components $C$ of $F_k$ with $\mdef(C)>0$ \\
        \For{$C \in \mathcal{C}$}{
            {\tt Search}($C$) \\
            if $C$ was merged with some $B \in \mathcal{C}$, delete $B$ from $\mathcal{C}$
        }
    }
    
    \caption{Solving Problem \ref{the_one_problem} for a given graph $G$ and value $k\ge 1$.}
    \label{alg_solve}
\end{algorithm}

The algorithm is correct because it keeps looking for an augmenting path until there are no augmenting paths.
If $F_k$ becomes a spanning tree, then we can stop: all forests must then be spanning trees, so $F$ is optimal.

\boundedIndegreeAlgorithm*
\begin{proof}
First, we consider the case $\tau = \mathbf{0}$ and show that Algorithm~\ref{alg_solve} can be implemented to run in time $O(k \delta m \log n)$.
There are two phases: growing forests for $k\le k_G$, and growing forests for $k>k_G$.
Below, we analyze the times for these operations excluding the time for computing indices $p$ in Algorithm~\ref{proc_augment}.

(1) {\em Growing the $k$-th forest for $k\le k_G$.} All forests $F_1,\ldots,F_{k-1}$ are spanning trees.
We claim that there will be no unsuccessful searches during this phase. 
    (If some search from node $v$ is unsuccessful, then Lemma~\ref{lem_q} gives an $F$-closed set $Q_v\ni v$ with  $Q_v\ne V$;
    the latter holds since $F_k$ is not a spanning tree.
    Eq.~\eqref{eq:Fclosed} gives $|\rho(Q_v)|=k-\tau(Q_v)-\Def_F(Q_v)<k\le k_G$, which is a contradiction.)
    Therefore, at each round of the \textbf{while} loop, the number of components of $F_k$ halves, so there will be $O(\log n)$ rounds of the loop. This yields  $O(m \log n)$ time for the $k$-th forest. Notice that this already settles the case of $k_G \geq k$, giving the desired bound $O(k m \log n)$, since $\delta = 1$ in this case.
    
(2) {\em Growing the $k$-th forest for $k> k_G$.}
At line 4, each component $C\in\mathcal{C}$ satisfies $\mdef(C)= \Def_F(C)\le k$ (by Theorem~\ref{theorem_low_deficit}),
    implying that $|\mathcal{C}|\ge \frac{\mdef(V)}{k}$.
    At least half of the components in $\mathcal{C}$ will be processed by {\tt Search}, so each round of the \textbf{while} loop decreases $\mdef(V)$ by at least $\frac{\mdef(V)}{2k}$.
    Therefore, there will be at most $- \log_{1 - \frac{1}{2k}} (k n)= O(k \log (k n))=O(k\log n)$ such rounds
    (since $k=O({\tt poly}(n))$, as assumed in Section~\ref{sec:intro}).
    Each round takes $O(\sum_{C\in\mathcal{C}}|E(C)|)=O(m)$ time.
    This sums to $O(km\log n)$ time for growing the $k$-th forest (lines 2-7).

In total, the operations above take $k_G\cdot O(m\log n) + \delta\cdot O(k m\log n)=O(k\delta  m \log n)$ time for the final value of $k$.
It remains to discuss the time for computing indices $p=\min\{j \; | \; e_r \text{ is joining for } F_j \}$
during augmentations. There are $k_G (n-1)$ augmentations in the first phase and
$O(\delta n)$ augmentations in the second phase.
In the first phase, we need $O(1)$ time per augmentation (since we always have $p=k$).
In the second phase, for each augmentation, we need to perform
the following operation $k$ times\footnote{With a binary search, one could improve this to $O(\log k)$ times, but the looser bound suffices for our purposes.}: determine whether a given edge $e$ is joining for a given forest $F_j$. If we maintain a disjoint-set data structure for each forest, then each such operation takes $O(\alpha(n))$ amortized time~\cite{union_find}, where $\alpha(\cdot)$ is the inverse Ackermann function.
In total, this yields $k_G(n-1)\cdot O(1) + O(\delta n)\cdot O(k\alpha( n))$ time, which is dominated by $O(k\delta m\log n)$.

Finally, consider the case $\tau = \tau^{a:k}$. Delete all edges entering $a$ and apply Algorithm~\ref{alg_solve} with $\tau={\bf 0}$ to the resulting graph. The feasible solutions are exactly those of the original problem with $\tau=\tau^{a:k}$, so the algorithm is correct. Now we analyze the runtime.
For $k\le k_G$, the algorithm works exactly like Gabow's algorithm and thus takes $O(k_G m\log n)$ time.
For $k>k_G$, Theorem~\ref{theorem_low_deficit} still applies to the graph with the edges entering $a$ deleted, so the same analysis yields time $O(k m \log n)$ for growing the next forest. Thus we have the same total time $O(k\delta m\log n)$.
\end{proof}

\subsection{Proof of Theorem \ref{theorem_low_deficit}}\label{appendix_low_deficit}

%This proof deserves to occupy a separate appendix. 
In this section, $\tau = {\bf 0}$.

Let us fix forests $\widetilde{F} = \widetilde{F}_1 \sqcup \ldots \sqcup \widetilde{F}_{k-1}$ that form an optimal solution for $k - 1$, and let $\calS=\{S_1,\ldots,S_q\}$ be the family of corresponding $\widetilde F$-closed sets constructed in Lemma~\ref{lem_q}.
We denote $S =  S_1\cup\ldots\cup S_q$ and $\overline S=V-S$.
In the rest of the proof, we also let $F = F_1 \sqcup \ldots \sqcup F_k$
 be a solution for $k$ (not necessarily optimal) that is obtained from $\widetilde{F}$ via a sequence of augmentations performed by \textnormal{\texttt{Search}}.
We will use the following notation:
\begin{itemize}
\item Let $F_{1 \ldots k-1} = F_1 \sqcup \ldots \sqcup F_{k-1}$.
\item Let $\C{F_k}$ be the set of components of the forest $F_k$; it forms a partition of $V$.
\item Let $\CC{F_k}=\{C\in\C{F_k}\::\:C\subseteq \overline S\}$ be the set of components in $\C{F_k}$ not intersecting $S$. 
\item Finally, define $\CCd{F_k}=\{C\in \CC{F_k}\::\:\Def_F(C)=d\}$.
\end{itemize}

\begin{lemma}\label{lem_F}
The solution $F=F_1\sqcup\ldots\sqcup F_k$ satisfies the following:
\begin{itemize} 
\item[(a)] $\indeg_{F_{1\ldots k-1}}(C)\ge \indeg_{\widetilde F}(C)$ for every $C\in\C{F_k}$.\footnote{The $\indeg$ function for a set of nodes is the total $\indeg$ of these nodes.}
\item[(b)] Each $C\in \CC{F_k}$ satisfies $\Def_F(C)\le 1$ (and hence $\CC{F_k}=\CCzero{F_k}\cup \CCone{F_k}$).
\item[(c)] $|F_{1\ldots k-1}|\ge|\widetilde F|+|\CCzero{F_k}|$.
\end{itemize}
\end{lemma}
\begin{proof}
We use induction on $|F|$. After initialization, we have $F_{1\ldots k-1}=\widetilde F$ and $F_k=\varnothing$.
By construction, $\CC{F_k}=\{\{v\}\::\:v\in\overline S\}$, every $\{v\}\in \CC{F_k}$
satisfies $\Def_F(v)=k-(k-1)=1$, and $\CCzero{F_k}=\varnothing$, implying (a)-(c).
Now suppose that the claim holds for $F$, and a solution $F'$ is obtained from $F$
by a single augmentation that starts at a vertex $v$ of a component $C\in\C{F_k}$
and ends at an edge $e$ which is joining for the forest indexed by $p\in[k]$.
Note that we have $\indeg_{F'}(v)=\indeg_{F}(v)+1$ and $\indeg_{F'}(u)=\indeg_{F}(u)$ for $u\in V-\{v\}$.
This means, in particular, that 
components $C'\in \C{F'_k}$ not containing $v$
have the same deficits in $F$ and $F'$.
Also, $\Def_F(C)\ge \Def_F(v)\ge 1$.
By Theorem~\ref{th:proc_augment}(a), $|F'_p|=|F_p|+1$ and $|F'_i|=|F_i|$ for $i\in[k]-p$.
Two cases are possible.
\begin{itemize}
\item $p<k$. Then we have $\C{F'_k}=\C{F_k}$,
 $\Def_{F'}(C)=\Def_F(C)-1$, and $\mbox{$|F'_{1\ldots k-1}|=|F_{1\ldots k-1}|+1$}$.
Note that $|\CCzero{F'_k}|\le|\CCzero{F_k}|+1$.
It can now be checked that~(b,c) cannot become violated for $F'$.
Only edges in $E(C)\cup\{e\}$ could have changed their memberships; therefore, $\indeg_{F'_{1\ldots k-1}}(C)-\indeg_{F_{1\ldots k-1}}(C)\ge |F'_{1\ldots k-1}|-|F_{1\ldots k-1}|-1\ge 0$,
which implies property~(a) for $F'$.

\item $p=k$. Since the terminal edge has exactly one endpoint in $C$, the augmentation merges $C$ with another component $D\in \C{F_k}-\{C\}$ into $\mbox{$C'=C\cup D$}$.
We have $\C{F'_k}=(\C{F_k}-\{C,D\})\cup\{C'\}$,
$\Def_{F'}(C')=\Def_F(C)+\Def_F(D)-1$, and 
  $|F'_{1\ldots k-1}|=|F_{1\ldots k-1}|$.
Only edges inside $E(C')$ could have changed their memberships; therefore, $\indeg_{F'_{1\ldots k-1}}(C')-\indeg_{F_{1\ldots k-1}}(C')=|F'_{1\ldots k-1}|-|F_{1\ldots k-1}|= 0$.
Using the induction hypothesis, we obtain
$\indeg_{F'_{1\ldots k-1}}(C')
=\indeg_{F_{1\ldots k-1}}(C)+\indeg_{F_{1\ldots k-1}}(D)
\ge \indeg_{\widetilde F}(C)+\indeg_{\widetilde F}(D)
=\indeg_{\widetilde F}(C')
$, which shows~(a). To show (b,c), we consider two subcases:

\underline{Case 1}: $C'\in \CC{F'_k}$, or equivalently $C'\subseteq \overline S$.
We then have $C,D\in \CC{F_k}$, and so the induction hypothesis yields $\Def_F(C)=1$ and $\Def_F(D)\le 1$.
Thus, (b) still holds (since $\Def_{F'}(C')=\Def_F(D)\le 1$),
and so does (c) (since $|F'_{1\ldots k-1}|=|F_{1\ldots k-1}|$ and $|\CCzero{F'_k}|=|\CCzero{F_k}|$).

\underline{Case 2}: $C'\notin \CC{F'_k}$. Then (b) cannot become violated, and (c) also cannot become violated
since $|F'_{1\ldots k-1}|=|F_{1\ldots k-1}|$ and $|\CCzero{F'_k}| \le |\CCzero{F_k}|$.\qedhere

\end{itemize}
\end{proof}

\begin{lemma}
Each component $C\in\C{F_k}$ intersects at most one of the sets $S_i\in\calS$.
    \label{lem_components_dont_connect_s}
\end{lemma}
\begin{proof}
Let $t=|\{(S_i,C)\in\calS\times\C{F_k} \::\:S_i\cap C\ne\varnothing\}|$ be the number of non-empty intersections between sets $S_i$ and
components of $F_k$.
 Clearly, we have $|\C{F_k}| \le |\CC{F_k}| + t$. Furthermore, if the lemma is false (i.e.,\ some $C\in\C{F_k}$ 
 has non-empty intersections with two distinct sets $S_i,S_{i'}\in\calS$)
 then $|\C{F_k}| \le |\CC{F_k}| + (t-1)$. To show the lemma, it thus suffices to prove that  $|\C{F_k}| \ge |\CC{F_k}| + t$.
 
   Let us split the edges into three sets:
    \begin{equation}
        \begin{cases}
            E_1 = E(S_1)\cup\ldots\cup E(S_q), \\
            E_2 = \rho(S_1)\cup\ldots\cup\rho(S_q), \\
            E_3 = \{e \in E \::\: \head(e) \in \overline{S}\}.
        \end{cases}
    \end{equation}
    Clearly, $E = E_1 \sqcup E_2 \sqcup E_3$.
    Denote $\Delta_i = |F \cap E_i| - |\widetilde{F} \cap E_i|$. Then $|F|-|\widetilde F|=\Delta_1+\Delta_2+\Delta_3$.
We claim that the quantities $\Delta_i$ can be upper-bounded as follows.
\begin{itemize}
\item $\Delta_1\le |S|-t$.
Indeed, for each $S_i$ we have $|F_k \cap E(S_i)| \le |S_i| - t_i$ where $t_i$ is the number of components $C\in\C{F_k}$ that intersect $S_i$.
Furthermore, $|F_{1\ldots k-1} \cap E(S_i)|\le  |\widetilde F \cap E(S_i)|$ since
$\widetilde{F}_1, \ldots, \widetilde{F}_{k-1}$ are already spanning in $S_i$.
This yields the desired claim as follows:
$$
\Delta_1=\sum_i\left(|F_k \cap E(S_i)|+|F_{1\ldots k-1} \cap E(S_i)|-|\widetilde F \cap E(S_i)| \right)
\le \sum_i\left( |S_i|-t_i\right) = |S| - t
$$
\item $\Delta_2\le 0$. This holds since each $S_i$ is $\widetilde F$-closed and hence all edges in $E_2$ are already covered by $\widetilde F$.
\item $\Delta_3 \leq |\overline{S}| - r + |F_{1 \ldots k-1}| - |\widetilde{F}|$.
Indeed,
for $v\in\overline S$, let us denote $\Delta^v_3 = \indeg_F(v) - \indeg_{\widetilde{F}}(v)$. Then $\Delta^v_3\le k - (k-1)=1$.
Each component $C\in\CCone{F_k}$ contains
$|C|-1$ nodes $v$ with $\Delta^v_3=1$ and one node $v$ with $\Delta^v_3=0$.
Using Lemma~\ref{lem_F}(b,c), we can now obtain the desired claim as follows:
$$
\Delta_3=\sum_{v\in\overline S}\Delta^v_3
\le |\overline S|- |\CCone{F_k}|
= |\overline S| - |\CC{F_k}| + |\CCzero{F_k}|
\le  |\overline S| - |\CC{F_k}| + |F_{1\ldots k-1}|-|\widetilde F|
$$
\end{itemize}

Putting together the bounds above gives
\begin{align*}
|F|-|\widetilde{F}|=\Delta_1+\Delta_2+\Delta_3 
&\le (|S|-t) + 0 + (|\overline S|-|\CC{F_k}|+|F_{1 \ldots k-1}| - |\widetilde{F}|) \\
&=|V|-t-|\CC{F_k}|+|F_{1\ldots k-1}|-|\widetilde{F}|
\end{align*}
and therefore
$
|F_k|=|F|-|F_{1\ldots k-1}|\le |V|-t-|\CC{F_k}|
$.
On the other hand, $|F_k|=|V|-|\C{F_k}|$, and hence  $|\C{F_k}|\ge |\CC{F_k}|+t$. This concludes the proof.
\end{proof}

Now we are ready to prove Theorem \ref{theorem_low_deficit}.

\begin{proof}[Proof of Theorem \ref{theorem_low_deficit}]
    Assume that there exists a component $C\in\C{F_k}$ with $\Def_F(C)>k$,
    or equivalently $\indeg_{F}(C)< k|C|-k$. Since there are $|C| - 1$ edges of $F_k$ in this component,
    we have $\indeg_{F_k}(C)= |C|-1$.
    Using Lemma~\ref{lem_F}(a), we conclude that $\indeg_{\widetilde F}(C)\le \indeg_{F_{1\ldots k-1}}(C)
    <(k|C|-k)-(|C|-1)=(k-1)|C|-k+1$ and hence $\Def_{\widetilde F}(C)>k-1$.
 By Lemma \ref{lem_components_dont_connect_s}, there exists $S_i\in\{S_1,\ldots,S_q,\varnothing\}$ such
 that each vertex from $C - S_i$ has zero deficit in $\widetilde F$, 
 and so $\Def_{\widetilde F}(S_i) = \Def_{\widetilde F}(C) > k - 1$.
 By the definition of $\calS$, we have $S_i\subseteq \widetilde C$ for some component $\widetilde C$ of $\widetilde F_{k-1}$.
 We obtain that $\Def_{\widetilde F}(\widetilde C)\ge \Def_{\widetilde F}(S_i)>k-1$, which is a contradiction.
\end{proof}

%%%%%%%%%%%%%%%%%%%%%%%%%%%%%%%%%%%%%%%%%%%%%%%%%%%%%%%%%%%%%%%%%%%%%%%%%%%%%%%%%%%%%%%%%%%%%%%%%%%%%%%%%%%%%%%%%%%%%%%%%%%%%%%%%%%%%%%%%%%%%%%
%%%%%%%%%%%%%%%%%%%%%%%%%%%%%%%%%%%%%%%%%%%%%%%%%%%%%%%%%%%%%%%%%%%%%%%%%%%%%%%%%%%%%%%%%%%%%%%%%%%%%%%%%%%%
%%%%%%%%%%%%%%%%%%%%%%%%%%%%%%%%%%%%%%%%%%%%%%%%%%%%%%%%%%%%%%%%%%%%%%%%%%%%%%%%%%%%%%%%%%%%%%%%%%%%%%%%%%%%
%%%%%%%%%%%%%%%%%%%%%%%%%%%%%%%%%%%%%%%%%%%%%%%%%%%%%%%%%%%%%%%%%%%%%%%%%%%%%%%%%%%%%%%%%%%%%%%%%%%%%%%%%%%%
%%%%%%%%%%%%%%%%%%%%%%%%%%%%%%%%%%%%%%%%%%%%%%%%%%%%%%%%%%%%%%%%%%%%%%%%%%%%%%%%%%%%%%%%%%%%%%%%%%%%%%%%%%%%
%%%%%%%%%%%%%%%%%%%%%%%%%%%%%%%%%%%%%%%%%%%%%%%%%%%%%%%%%%%%%%%%%%%%%%%%%%%%%%%%%%%%%%%%%%%%%%%%%%%%%%%%%%%%
%%%%%%%%%%%%%%%%%%%%%%%%%%%%%%%%%%%%%%%%%%%%%%%%%%%%%%%%%%%%%%%%%%%%%%%%%%%%%%%%%%%%%%%%%%%%%%%%%%%%%%%%%%%%
%%%%%%%%%%%%%%%%%%%%%%%%%%%%%%%%%%%%%%%%%%%%%%%%%%%%%%%%%%%%%%%%%%%%%%%%%%%%%%%%%%%%%%%%%%%%%%%%%%%%%%%%%%%%
%%%%%%%%%%%%%%%%%%%%%%%%%%%%%%%%%%%%%%%%%%%%%%%%%%%%%%%%%%%%%%%%%%%%%%%%%%%%%%%%%%%%%%%%%%%%%%%%%%%%%%%%%%%%
%%%%%%%%%%%%%%%%%%%%%%%%%%%%%%%%%%%%%%%%%%%%%%%%%%%%%%%%%%%%%%%%%%%%%%%%%%%%%%%%%%%%%%%%%%%%%%%%%%%%%%%%%%%%
%%%%%%%%%%%%%%%%%%%%%%%%%%%%%%%%%%%%%%%%%%%%%%%%%%%%%%%%%%%%%%%%%%%%%%%%%%%%%%%%%%%%%%%%%%%%%%%%%%%%%%%%%%%%
%%%%%%%%%%%%%%%%%%%%%%%%%%%%%%%%%%%%%%%%%%%%%%%%%%%%%%%%%%%%%%%%%%%%%%%%%%%%%%%%%%%%%%%%%%%%%%%%%%%%%%%%%%%%

\section{Algorithm for Problem~\ref{problem_k_forest}}\label{sec:problem_k_forest}
Besides the algorithm for Problem~\ref{the_one_problem}, we will need two more subroutines.
They are described in Sections~\ref{subsection_pseudoforests} and~\ref{sec:clump}.

\subsection{Procedure {\tt Pseudoforests} %A maximum subgraph admitting the indegree-at-most-$k$-orientation
}\label{subsection_pseudoforests}
In this subsection, let $G$ be an undirected graph with $n$ vertices and $m$ edges. Let $H$ be a subgraph of $G$. Moreover, suppose that $H$ has been oriented in such a way that $\indeg_H(v) \leq k$ for any vertex $v$. Consider the following problem: find a subgraph $P \subseteq G$ with the maximum number of edges such that $H \subseteq P$, and the edges of $P$ can be oriented in such a way that $\indeg_P(v) \leq k$ for any vertex $v$. Output $P$ together with its orientation. The orientation on $P$ does not necessarily induce the same orientation on $H$ as $H$ had in the input. We will call this the pseudoforest problem, since a graph $P$ can be oriented such that $\indeg_P(v) \leq k$ for any vertex $v$ if and only if $P$ can be covered by $k$ pseudoforests. The problem can be solved with a single maximum flow computation on a specially constructed graph that appeared in {\cite{gabow_polymatroids}}.

\begin{theorem}\label{th_pseudoforests_via_maxflow}
    There exists an algorithm {\tt Pseudoforests}$(G, H, k)$ that solves the pseudoforest problem in $O({\rm MAXFLOW}(m, m))$ time. Here, ${\rm MAXFLOW}(n, m)$ is the complexity of a max flow computation in a graph with $n$ vertices and $m$ edges.
\end{theorem}
\begin{proof}
    Fix $k$ and $G$. We are going to define an auxiliary graph $G^*$. This definition appeared in \cite{gabow_polymatroids}.

    \begin{definition}
        $G^*$ is a directed weighted graph with vertex set $\{s\} \sqcup E \sqcup V \sqcup \{t\}$ and the following edges:
        \begin{enumerate}
            \item $(s, e)$ for all $e \in E$, capacity 1,
            \item $(e, v), (e, u)$ for all $e = (v, u) \in E$, $v, u \in V$, capacity 1,
            \item $(v, t)$ for all $v \in V$, capacity $k$.
        \end{enumerate}
    \end{definition}
    
    By construction, $G^*$ has $O(m)$ vertices and $O(m)$ edges.

    % First, notice that any integer $s-t$ flow in $G^*$ encodes a subgraph with an orientation such that the indegrees of vertices are at most $k$ and vice versa. 
    First, notice that there is a one-to-one correspondence between integer $s-t$ flows in $G^*$ and subgraphs with an orientation such that the indegrees of vertices are at most $k$. Indeed, let the subgraph consist of those edges $e$ with positive flow through $(s, e)$. The orientation of $e = (u, v)$ is determined as follows: if flow passes through $(e, u)$, then $e$ gets oriented from $v$ to $u$, and if flow passes through $(e, v)$, then $e$ gets oriented from $u$ to $v$. The flow through edges of type $(v, t)$ equals the indegree of $v$ under this orientation. Therefore, the capacity $k$ of the edges of type $(v, t)$ ensures that the indegrees of vertices are at most $k$. In the same way, one can construct a flow in $G^*$ given a subgraph oriented subject to the indegree condition.

    Take the oriented subgraph $H$ and send the corresponding flow through $G^*$. After that, run an integer max flow algorithm
    \footnote{
If the max flow routine outputs a fractional flow, it can be rounded to an integer flow with the same flow value in almost linear time \cite{flowrounding}.}
through the residual graph, obtaining a max flow in $G^*$. Let $P$ be the oriented subgraph encoded by the resulting flow. We make two claims.

    First, $P$ is a maximum subgraph satisfying the indegree condition. This is because the value of the flow equals the number of edges in $P$, and thus $|P|$ is maximum since the flow is maximum.

    Second, $H \subseteq P$. This is because the residual edges of the form $(e, s)$ for $e \in H$ are not used in the flow through the residual graph.\footnote{To ensure that these edges are not used, one can delete these edges from the residual graph before running the max flow. It will not affect the optimal value of the flow.} Therefore, each such edge $(s, e)$ will still carry flow in the final flow.

    Together, these two claims show that the described procedure solves the pseudoforest problem.
\end{proof}

\subsection{Contractions and top clumps}\label{sec:clump}

In this subsection, $G=(V,E)$ is once again undirected. 
Our algorithm will rely on the operation of {\em contraction}, as described in the following lemma.

\begin{lemma}\label{lemma:contract}
Let $F_1,\ldots,F_k$ be a feasible solution of the $k$-{\sc forest} problem for $G$
and $U$ be a subset of $V$ such that each $T_i=F_i\cap E(U)$
is a spanning tree of $(U, E(U))$.
Let $G'  = (V', E')$ be the graph obtained from $G$ by contracting
$U$ to a single node and deleting self-loops. \\
\noindent {\rm (a)} $F_1-E(U),\ldots,F_k-E(U)$ is a feasible solution for $G'$. \\
\noindent {\rm (b)} If $F'_1,\ldots,F'_k$ is an optimal solution for $G'$, then $F'_1\sqcup T_1,\ldots,F'_k\sqcup T_k$ is an optimal solution for~$G$.\!\!\!\!
\end{lemma}
\noindent Note that in the lemma we treat edges $E'$ of $G'$ 
as a subset of edges $E$ of $G$ via a natural one-to-one correspondence between $E-E(U)$ and $E'$.

\begin{proof}
Part {\rm (a)} is straightforward, as is the feasibility of the solution $F'_1\sqcup T_1,\ldots,F'_k\sqcup T_k$
in {\rm (b)}. We need to prove the optimality of the latter solution.

    Denote $F' = F'_1 \sqcup \ldots \sqcup F'_k$, $T = T_1 \sqcup \ldots \sqcup T_k$. Let $\mathcal{G}$ and $\mathcal{G}'$ be the graphic matroids corresponding to $G$ and $G'$, respectively. Using the matroid union theorem,
    \begin{equation*}
        |F'| = \min_{H \subseteq E'} \{ |E' - H| + k \cdot \text{rk}_{\mathcal{G}'}(H) \}.
    \end{equation*}
    Let $H^* \subseteq E'$ be a set that attains the minimum. Let $c$ be the number of connected components of $(V',H^*)$. Then $(V,H^*\cup E(U))$ also has $c$ connected components, since $G[U]$ is connected. Thus, $\text{rk}_{\mathcal G}(H^*\cup E(U))=n-c$. Therefore,
    \begin{equation}
        |F'| = |E' - H^*| + k \cdot \text{rk}_{\mathcal{G}'}(H^*) = |E' - H^*| + k (n - |U| + 1 - c).
    \end{equation}
    
    Now use the matroid union theorem again for $\mathcal{G}^k$ (the union of $k$ matroids $\mathcal{G}$):
    \begin{align*}
            \text{rk}_{\mathcal{G}^k}(E) &= \min_{H \subseteq E} \{ |E - H| + k \cdot \text{rk}_{\mathcal{G}}(H) \}  \\
            &\leq |E - H^* - E(U)| + k \cdot \text{rk}_{\mathcal{G}}(H^* + E(U)) \leq |E' - H^*| + k (n - c)  \\
            &= |F'| + k (|U| - 1) = |F' + T|.
    \end{align*}
    Therefore, $F' \sqcup T$ is optimal.
\end{proof}

Next, we discuss how to compute sets $U$ for a given solution $F_1,\ldots,F_k$
that satisfy the precondition of Lemma~\ref{lemma:contract}.
For that, we will need the notion of {\em clumps}.
They were introduced in \cite{GabowWestermann:92} for arbitrary matroids. 
Here we specialize them for graphic matroids.
%However, in this paper, we are interested in uniting graphic matroids, so we are going to define clumps regarding graphic matroids only.

\begin{definition}
Consider an undirected graph $(V,F)$ and a positive integer $k$.
A set of edges $L\subseteq F$ is called a \textit{clump of $F$}
if $L$ can be covered by $k$ forests and 
$|L| = k \cdot \text{rk}(L)$, where $\text{rk}(L)$ is the number of edges in a spanning forest in $L$.
Equivalently, $L$ is a clump %if and only 
if each connected component of $L$ can be partitioned into $k$ spanning trees.
A {\em top clump of $F$} is an inclusion-maximal clump of~$F$.
\end{definition}

Our algorithm will use the following operation:
given a current feasible solution $F_1,\ldots,F_k$ in $G$,
compute a top clump $L$ of $F_1\sqcup\ldots\sqcup F_k$,
find connected components $U_1,\ldots,U_q$ of $L$,
and contract them consecutively as described in Lemma~\ref{lemma:contract}.
Luckily, this can be done efficiently due to results of Gabow and Westermann 
(see the paragraph before Theorem 5.2 and Theorem 5.4 in~\cite{GabowWestermann:92}).
\begin{lemma}[\cite{GabowWestermann:92}]
Suppose that $F=F_1\sqcup\ldots\sqcup F_k$ where $F_1,\ldots,F_k$ are edge-disjoint forests.\\
{\rm (a)} A top clump of $F$ is unique. \\
{\rm (b)} Given $F_1,\ldots,F_k$, the top clump of $F$ can be computed in $O(kn \log n)$ time.
\end{lemma}

We will need one more technical result. We say that a feasible solution $F=F_1\sqcup\ldots\sqcup F_k$
in $G$ is {\em clump-free} if the top clump of $F$ is empty. The next lemma will be needed in the analysis of the main algorithm.
\begin{lemma}\label{lemma:clump-free}
Let $G'$ and $F'$ be the graph and the solution obtained after contracting
all connected components $U_1,\ldots,U_q$ of the top clump $L$ of $F$
as described in Lemma~\ref{lemma:contract}.
Then $F'$ is clump-free.
\end{lemma}
\begin{proof}
    Assume the contrary. Then, there is $H \subseteq F'$ that consists of $k$ spanning trees $T_1, \ldots, T_k$. Uncontract the connected components of $L$ one by one. $L \sqcup H$ can still be decomposed into $k$ spanning forests; therefore, it is a clump. We have a contradiction with the maximality of $L$.
\end{proof}

%%%%%%%%%%%%%%%%%%%%%%%%%%%%%%%%%%%%%%%%%%%%%%%%%%%%%%%%%%%%%%%%%%%%%%%%%%%%%%%%%%%%%%%%%%%%%%%%%%%%%%%%%%%%%%
%%%%%%%%%%%%%%%%%%%%%%%%%%%%%%%%%%%%%%%%%%%%%%%%%%%%%%%%%%%%%%%%%%%%%%%%%%%%%%%%%%%%%%%%%%%%%%%%%%%%%%%%%%%%%%
%%%%%%%%%%%%%%%%%%%%%%%%%%%%%%%%%%%%%%%%%%%%%%%%%%%%%%%%%%%%%%%%%%%%%%%%%%%%%%%%%%%%%%%%%%%%%%%%%%%%%%%%%%%%%%
%%%%%%%%%%%%%%%%%%%%%%%%%%%%%%%%%%%%%%%%%%%%%%%%%%%%%%%%%%%%%%%%%%%%%%%%%%%%%%%%%%%%%%%%%%%%%%%%%%%%%%%%%%%%%%
%%%%%%%%%%%%%%%%%%%%%%%%%%%%%%%%%%%%%%%%%%%%%%%%%%%%%%%%%%%%%%%%%%%%%%%%%%%%%%%%%%%%%%%%%%%%%%%%%%%%%%%%%%%%%%
%%%%%%%%%%%%%%%%%%%%%%%%%%%%%%%%%%%%%%%%%%%%%%%%%%%%%%%%%%%%%%%%%%%%%%%%%%%%%%%%%%%%%%%%%%%%%%%%%%%%%%%%%%%%%%
%%%%%%%%%%%%%%%%%%%%%%%%%%%%%%%%%%%%%%%%%%%%%%%%%%%%%%%%%%%%%%%%%%%%%%%%%%%%%%%%%%%%%%%%%%%%%%%%%%%%%%%%%%%%%%
%%%%%%%%%%%%%%%%%%%%%%%%%%%%%%%%%%%%%%%%%%%%%%%%%%%%%%%%%%%%%%%%%%%%%%%%%%%%%%%%%%%%%%%%%%%%%%%%%%%%%%%%%%%%%%
%%%%%%%%%%%%%%%%%%%%%%%%%%%%%%%%%%%%%%%%%%%%%%%%%%%%%%%%%%%%%%%%%%%%%%%%%%%%%%%%%%%%%%%%%%%%%%%%%%%%%%%%%%%%%%
%%%%%%%%%%%%%%%%%%%%%%%%%%%%%%%%%%%%%%%%%%%%%%%%%%%%%%%%%%%%%%%%%%%%%%%%%%%%%%%%%%%%%%%%%%%%%%%%%%%%%%%%%%%%%%

\subsection{Main algorithm for Problem \ref{problem_k_forest}}\label{sec:main}

In this section, we describe a solution to the $k$-{\sc forest} problem 
that uses the algorithm for Problem~\ref{the_one_problem}, {\tt Pseudoforests} and top clump computations as subroutines
(see Algorithm~\ref{proc_k_forests}). Note that the orientation of $F$ at line 3
exists and can be computed efficiently: in each $F_i$
we can orient edges away from roots chosen arbitrarily in each component
(then $\indeg_{F_i}(v)\le 1$ for each $v$).
% We claim that the following algorithm solves the problem.

\begin{algorithm}[!htbp]
    \DontPrintSemicolon
    \SetNoFillComment

    $F_1, \ldots, F_k \gets \varnothing, \ldots, \varnothing$ \\
    \Repeat{$|F|$ does not change, counted in the uncontracted graph} {
        orient edges in $F$ so that $\indeg_F(v)\le k$ for each $v\in V$ 
        \\ %\hspace{200pt} \tcc{for example, choose a root in each component of each forest \hspace{20pt} and then orient edges away from the roots}
        %orient each $F_i$ such that $\indeg_{F_i} (v) \leq 1$ $\forall v \in V$ \\
        %% away from roots arbitratily chosen in each component (so $\indeg_F v \leq k$ $\forall v \in V$) \\
        $P \gets$ {\tt Pseudoforests}$(G, F, k)$ \\
        solve Problem~\ref{the_one_problem} for input $((V,P),k,{\bf 0})$, and let $F=F_1 \sqcup \ldots \sqcup F_k$ be an optimal solution \\
%        $F_1, \ldots, F_k \gets$ {\tt BoundedIndegreeForests}$(P, k)$ \\ 
        find the top clump of $F$ \\
        contract each component of the top clump into a supervertex and update $G,F$ accordingly
    }
    recursively uncontract the contracted components together with forests $F_1,\ldots,F_k$ (see Sec.~\ref{sec:clump})\!\!\!\!\!\!\!\!\!\!\!\! \\
    \Return $F_1, \ldots, F_k$
    
    \caption{{\tt Forests}($G$: undirected graph, $k$: positive integer).}
    \label{proc_k_forests}
\end{algorithm}

% In the stopping condition, the solution size includes the covered edges stored in all contracted components. Contraction preserves this size; only the computation of a new forest packing can increase it.

For a graph $G$, let ${\tt OPT}(G)$ 
be the size of an optimal solution to the $k$-{\sc forest} problem for $G$.
The following lemma analyzes 
%We start with the analysis of 
a single iteration of the {\bf repeat} loop (lines 3-7).
\begin{lemma}\label{lemma:main-analysis}
Let $(G,F)$ be the graph and the solution at the beginning of an iteration,
$H=H_1\sqcup\ldots\sqcup H_k$ be the solution computed at line 5, 
and $(G',F')$ be the graph and the solution after contraction. Then
\begin{equation}\label{eq:diff-change}
{\tt OPT}(G')-|F'|
\;\;=\;\;{\tt OPT}(G)-|H|
\;\;\le\;\; \tfrac{k}{k+1} \cdot ({\tt OPT}(G)-|F|)
\end{equation}
\end{lemma}

\begin{proof}
The first equality in~\eqref{eq:diff-change} follows directly from Lemma~\ref{lemma:contract}.
We thus focus on proving the last inequality in~\eqref{eq:diff-change}.

For the analysis, we can assume w.l.o.g.\ that the orientation of edges in $F$ at line 3
coincides with the orientation in $P$. This is a valid orientation since $\indeg_F(v)\le \indeg_P(v)\le k$ for each $v$,
and reorienting edges of $F$ does not affect the output of {\tt Pseudoforests}$(G, F, k)$ at line 4.

Note that at line 5
the bounded degree constraint in Problem~\ref{the_one_problem} will never be active, since ${\tt indeg}_P(v)\le k$ for all $v\in V$.
Therefore, this procedure is equivalent to solving the $k$-{\sc forest} problem in the undirected version of the graph $(V,P)$.
$F$ is a feasible solution, so by a property of matroids, there exists an optimal solution containing $F$.
For the purpose of proving the last inequality in~\eqref{eq:diff-change}
we can thus assume w.l.o.g.\ that $F\subseteq H \subseteq P$.
%(viewed as undirected edges).
Denote $\overline H=P-H$, %, so that $P=F\sqcup A \sqcup B$ (ignoring edge orientations).
and    let $\calS = \{S_i\}_{i = 1}^q$ be the sets constructed in Lemma~\ref{lem_q}(b)
    for the optimal solution $H$ in $P$. 
    Note that each $S\in\calS$ is $H$-closed (as in Definition~\ref{def:closed}).
    We make the following claims.
    \begin{itemize}
    \item $\overline H\subseteq\bigcup_{S\in\calS}E(S)$. Indeed, consider $e=(u,v)\in \overline H$.
    We have $\indeg_H(v)\le \indeg_P(v)-1\le k-1$, so by condition (ii) of Lemma~\ref{lem_q}(b)
    there exists $S\in\calS$ with $v\in S$.
    $H$-closedness of $S$ now implies that $u\in S$.
    \item For each $S\in\calS$ with $|S|\ge 2$, $E(S)$ contains at least one edge of $H-F$. Indeed, otherwise $F\cap E(S)=H\cap E(S)$, since $F\subseteq H$. Since $S$ is $H$-closed, the set $F\cap E(S)$ has $k(|S|-1)$ edges and is a nonempty clump of $F$. However, the previous iteration contracted the top clump, so $F$ is clump-free by Lemma~\ref{lemma:clump-free}, which is a contradiction.
    \item For each $S\in\calS$, $E(S)$ cannot contain more than $k$ edges of $\overline H$.
    Indeed, we have $\indeg_P(S) \leq k |S|$ and also $\indeg_{H}(S) \ge k |S| - k $
	since $S$ is $H$-closed. Therefore, 
%	The claim now follows from
    \begin{equation*}
     \indeg_{\overline H}(S)=\indeg_P(S) - \indeg_{H}(S)
     \le k|S| - (k|S|-k)=k
    \end{equation*}
    which implies the claim.
    \end{itemize}
    The last two facts imply that $|\overline H\cap E(S)|\le k\cdot|(H-F)\cap E(S)|$ for each $S\in\calS$ with $|S|\ge 2$. For singleton sets $S$, the same inequality holds because $E(S)=\varnothing$.
    Summing this over $S\in\calS$ and using the first fact gives $|\overline H|\le k\cdot |H-F|$,
    or equivalently $|P|-|H|\le k (|H| - |F|)$. 
    
$H$ is a feasible solution to the $k$-{\sc forest} problem in $G$, so
by a property of matroids, there must exist an optimal solution $F^\ast$ (with $|F^\ast|={\tt OPT}(G)$)
such that $H\subseteq F^\ast$. The set of edges $F^\ast$ with an appropriate orientation
is a feasible solution to the pseudoforest problem at line 4; therefore, ${\tt OPT}(G)=|F^\ast|\le|P|$. Combining it with the previous inequality gives
${\tt OPT}(G)-|H|\le k (|H| - |F|)$.
Finally, adding a constant to both sides yields
    \begin{equation*}\label{eq:kPH}
     (k+1)({\tt OPT}(G)-|H|) \le k({\tt OPT}(G)-|F|)
     \qquad\Longrightarrow\qquad
     {\tt OPT}(G)-|H|\le \tfrac{k}{k+1}\cdot({\tt OPT}(G)-|F|)
    \end{equation*}
\end{proof}

\undirectedForestsAlgorithm*
\begin{proof}
We show that Algorithm {\tt Forests} solves Problem~\ref{problem_k_forest} within the claimed time bound.
By Lemma~\ref{lemma:main-analysis}, there will be at most
    $
        O(\log_{\tfrac{k+1}{k}} (k n)) = O(k \log (k n)) = O(k \log n)
    $ iterations of the {\bf repeat} loop, and we will have
 ${\tt OPT}(G)-|F|=0$ upon termination.
    The pseudoforest computation in line 4 runs in $O({\rm MAXFLOW}(m, m))$, as discussed in Subsection \ref{subsection_pseudoforests}. The size of $P$ is bounded by
    $\min\{k n,m\}$, so solving Problem~\ref{the_one_problem} at line 5 takes $O(k^2 \min\{k n,m\} \log n)$ time. This dominates the computation of the top clump at line 6. The final complexity is then $O(k \log n \cdot({\rm MAXFLOW}(m, m) + k^2\min\{k n,m\} \log n))$, as claimed.
\end{proof}

\begin{remark}
    One may use sparsification techniques to reduce the complexity. In \cite{quanrud}, it is stated (with a proof sketch) that if we greedily pack $\lceil k(1 + \log_2 n ) \rceil$ maximal forests in a graph and delete the remaining edges, this will not affect the size of the maximum $k$ forests (Lemma 8.5 in \cite{quanrud}, specialized to graphic matroids). Such a greedy packing can be found in linear time using the Nagamochi-Ibaraki technique \cite{NagamochiIbaraki_sparsification}. Thus, ``$m$'' in the runtime can be replaced by $\min\{k n \log n, m\}$ with an $O(m)$ overhead.
\end{remark}

%%%%%%%%%%%%%%%%%%%%%%%%%%%%%%%%%%%%%%%%%%%%%%%%%%%%%%%%%%%%%%%%%%%%%%%%%%%%%%%%%%%%%%%%%%%%%%%%%%%%%%%%%%%%
%%%%%%%%%%%%%%%%%%%%%%%%%%%%%%%%%%%%%%%%%%%%%%%%%%%%%%%%%%%%%%%%%%%%%%%%%%%%%%%%%%%%%%%%%%%%%%%%%%%%%%%%%%%%
%%%%%%%%%%%%%%%%%%%%%%%%%%%%%%%%%%%%%%%%%%%%%%%%%%%%%%%%%%%%%%%%%%%%%%%%%%%%%%%%%%%%%%%%%%%%%%%%%%%%%%%%%%%%
%%%%%%%%%%%%%%%%%%%%%%%%%%%%%%%%%%%%%%%%%%%%%%%%%%%%%%%%%%%%%%%%%%%%%%%%%%%%%%%%%%%%%%%%%%%%%%%%%%%%%%%%%%%%
%%%%%%%%%%%%%%%%%%%%%%%%%%%%%%%%%%%%%%%%%%%%%%%%%%%%%%%%%%%%%%%%%%%%%%%%%%%%%%%%%%%%%%%%%%%%%%%%%%%%%%%%%%%%
%%%%%%%%%%%%%%%%%%%%%%%%%%%%%%%%%%%%%%%%%%%%%%%%%%%%%%%%%%%%%%%%%%%%%%%%%%%%%%%%%%%%%%%%%%%%%%%%%%%%%%%%%%%%
%%%%%%%%%%%%%%%%%%%%%%%%%%%%%%%%%%%%%%%%%%%%%%%%%%%%%%%%%%%%%%%%%%%%%%%%%%%%%%%%%%%%%%%%%%%%%%%%%%%%%%%%%%%%
%%%%%%%%%%%%%%%%%%%%%%%%%%%%%%%%%%%%%%%%%%%%%%%%%%%%%%%%%%%%%%%%%%%%%%%%%%%%%%%%%%%%%%%%%%%%%%%%%%%%%%%%%%%%
%%%%%%%%%%%%%%%%%%%%%%%%%%%%%%%%%%%%%%%%%%%%%%%%%%%%%%%%%%%%%%%%%%%%%%%%%%%%%%%%%%%%%%%%%%%%%%%%%%%%%%%%%%%%
%%%%%%%%%%%%%%%%%%%%%%%%%%%%%%%%%%%%%%%%%%%%%%%%%%%%%%%%%%%%%%%%%%%%%%%%%%%%%%%%%%%%%%%%%%%%%%%%%%%%%%%%%%%%
%%%%%%%%%%%%%%%%%%%%%%%%%%%%%%%%%%%%%%%%%%%%%%%%%%%%%%%%%%%%%%%%%%%%%%%%%%%%%%%%%%%%%%%%%%%%%%%%%%%%%%%%%%%%

\section{Algorithm for Problem~\ref{problem_augmentation}}\label{sec:directed}

Recall that by Theorem~\ref{theorem_frank}, the minimum number of directed edges that should be added to $G$ to make it strongly $k$-connected
equals
$\gamma(G, k) = \max \left\{ \alpha_{\tt in}(G, k), \alpha_{\tt out}(G, k) \right\}$ where

    \begin{eqnarray}
	\alpha_{\tt in}(G, k)&=&\max_{\raisebox{-5pt}{\mbox{\small $\calX$: a proper subpartition of $V$}}} \;\; \sum_{A\in\calX}(k-|\rho(A)|)     \end{eqnarray}
and	
$\alpha_{\tt out}(G, k)=\alpha_{\tt in}(G^{\tt rev}, k)$.
We will need a few definitions.
\begin{definition}\label{def:ks-connected}
    A graph $G' = (V + s, E')$ is \textit{$(k,s)$-connected} if, for every pair of vertices $u, v \in V$, the size of a minimum $u$-$v$ cut is at least $k$.
    Such a graph $G'$ is a \textit{$(k,s)$-connected extension of a graph $G=(V,E)$} if $G'$ is $(k,s)$-connected, $E\subseteq E'$, and all edges in $E'-E$ are incident to $s$.
%    A graph $G' = (V + s, E')$, where $E'$ consists of $E$ and possibly some edges incident to $s$, is called a $(k, s)$-\textit{connected extension} of $G$ if for every pair of vertices $u, v \in V$, the size of the minimal $u$-$v$ cut is at least $k$.
\end{definition}

\begin{theorem}[Frank, {\cite[Lemma 3.3]{frank}}]
    For a directed graph $G$, there exists a $(k, s)$-connected extension $G' = (V + s, E')$ with $\indeg(s) = \outdeg(s) = \gamma(G, k)$.
    \label{th_greedy_extension}
\end{theorem}

The next theorem was originally proved by Mader in the paper~\cite{mader_splitting}, which is written in German. For an English source, see \cite[Theorem 3.4]{frank}.
\begin{theorem}[Mader, \cite{mader_splitting}]
    Suppose that the graph $G' = (V + s, E')$ is $(k, s)$-connected and $\indeg(s) = \outdeg(s)$.
    Then, for every edge $(u, s)\in E'$, there exists an edge $(s, v)\in E'$ such that the graph $G''=(V+s,E'-\{(u,s),(s,v)\}+\{u,v\})$ is also $(k,s)$-connected.
\end{theorem}
The operation of replacing edges $(u,s),(s,v)$ with $(u,v)$
that preserves $(k,s)$-connectivity is called {\em edge splitting}. The theorems above imply
that Problem \ref{problem_augmentation} can be solved as follows:
\begin{enumerate}
    \item Generate a $(k, s)$-connected extension $G' = (V + s, E')$ with $\indeg(s) = \outdeg(s) = \gamma(G, k)$.
    Such an extension will be called {\em optimal}.

    \item Perform edge splittings while $s$ is not isolated. This operation is called splitting off $s$.
\end{enumerate}

Below we will show that the first subproblem, computing $G'$, can be solved in $O(k\delta m\log n)$ time.
Let $T_{\rm split}$ denote the time needed to split off $s$ completely while preserving $(k,s)$-connectivity. This yields an algorithm for Problem~\ref{problem_augmentation} with overall complexity $O(k\delta m\log n)+T_{\rm split}$. The splitting algorithm of Bhalgat et al.~\cite[Theorem 1.2]{edge_splitting} splits off a vertex $s$ (in fact more general, any independent set of vertices) while preserving $r$-connectivity from a fixed root vertex $r$. This does not in general suffice to preserve $(k,s)$-connectivity. However, we show below that it does suffice when $G$ is Eulerian.

\paragraph{Computing an optimal $(k,s)$-connected extension $G'$ of $G$.}

% \subsection{Computing an optimal \texorpdfstring{$(k,s)$}{(k,s)}-connected extension \texorpdfstring{$G'$}{G'} of \texorpdfstring{$G$}{G}}\label{sec:optimal-extension}

%In this section, we describe an algorithm for computing a $(k,s)$-connected extension $G'$ of a given directed graph $G$
%with $\indeg(s) = \outdeg(s)=\gamma(G,k)=\max\{\alpha_{\tt in}(G,k),\alpha_{\tt out}(G,k)\}$; such an extension will be called {\em optimal}.
%Recall that given optimal $G'$ with $m'$ edges, Problem~\ref{problem_augmentation} can be solved in $O(m'k\log n)$ time
%using edge splitting algorithms of~\cite{edge_splitting} (see Section~\ref{sec:1.2}).

We say that a vector $\eta:V\rightarrow \mathbb Z_{\ge 0}$ is a {\em $k$-half-extension of $G$}
if 
\begin{equation}\label{eq:half-extension}
|\rho(A)|+\eta(A)\ge k\qquad\quad \forall A\subsetneq V, A\ne \varnothing
\end{equation}
It is {\em minimal} if there is no $k$-half-extension $\eta'\ne\eta$ 
with $\eta'\le \eta$. % (i.e.,\ with $\eta'(v)\le \eta(v)$ for all $v\in V$).
Equivalently, $\eta$ is a minimal $k$-half-extension if and only if it can be
obtained by the greedy algorithm that initializes $\eta=(k,\ldots,k)$ and then keeps doing the following step while possible:
\begin{itemize}
\item {\em pick some $v\in V$ with $\eta(v)>0$ and decrease $\eta(v)$ by 1 if this preserves~\eqref{eq:half-extension}.}
\end{itemize}

Frank showed that this greedy procedure can be used to construct an optimal $(k,s)$-connected extension
%in Lemma~\ref{lemma_greedy_extension} 
as follows.

\begin{theorem}[Frank, reformulation of {\cite[Theorem 3.9]{frank}}]
Let $\eta$ and $\eta^{\tt rev}$ be minimal $k$-half-extensions of $G$ and $G^{\tt rev}$, respectively.
Then a $(k,s)$-connected extension $G'$ of $G$ with ${\tt indeg}(s)={\tt outdeg}(s)=\gamma(G,k)$
can be obtained from $G$ as follows: \\
(i) add a new node $s$;  \\
(ii) for each $v\in V$ with $\eta(v)>0$, add $\eta(v)$ edges $(s,v)$; \\
(iii) for each $v\in V$ with $\eta^{\tt rev}(v)>0$, add $\eta^{\tt rev}(v)$ edges $(v,s)$; \\
(iv) if $\eta(V)>\eta^{\tt rev}(V)$ then add $\eta(V)-\eta^{\tt rev}(V)$ arbitrary edges of the form $(v,s)$,
otherwise add $\eta^{\tt rev}(V)-\eta(V)$ arbitrary edges of the form $(s,v)$.
\end{theorem}
If $G$ is Eulerian, then every cut has equally many incoming and outgoing edges. Consequently, we can choose $\eta^{\tt rev}=\eta$ in the construction above, obtaining an Eulerian optimal extension $G'$. Splitting off pairs of edges preserves Eulerianness. In an Eulerian graph, connectivity of at least $k$ from a fixed root $r\in V$ to every other vertex in $V$ implies $(k,s)$-connectivity, since every cut has at least $k$ edges in each direction. Thus, the splitting algorithm of Bhalgat et al.~\cite[Theorem 1.2]{edge_splitting}, that preserves connectivity to some root vertex $r$, can be used to split off $s$. Its running time is $O(km'\log n)$, where $m'\le m+2\delta n$ is the number of edges in $G'$. Since $m=\Omega(n)$ and $\delta\ge 1$, the total complexity in this case is dominated by constructing the extension and equals $O(k\delta m\log n)$.

Thus, it suffices to show how to compute a minimal $k$-half-extension of a given graph $G$.
The next two theorems give two approaches for solving this problem.
Note that the first approach is applicable only in a special case.
\begin{theorem}\label{th:approachA}
Let $F$ be an optimal solution of Problem~\ref{the_one_problem} for $(G,k,{\bf 0})$.
Then $\alpha_{\tt in}(G,k)\le \Def_F(V)$.
Now suppose that $\Def_F(V)>k$ (equivalently, not all forests in $F$ are spanning trees).
Then the vector $\eta$ with $\eta(v)=\Def_F(v)$ for $v\in V$ is a minimal $k$-half-extension of~$G$,
and furthermore $\alpha_{\tt in}(G,k)=\Def_F(V)>k$.
%Otherwise, if $\Def_F(V)\le k$, then $\alpha_{\tt in}(G,k)\le k$.
\end{theorem}
\begin{theorem}\label{th:approachB}
Fix a node $a\in V$, and let $F$ be an optimal solution of Problem~\ref{the_one_problem} for $(G,k,\tau^{a:k})$.
%, which can be obtained with the help of Observation \ref{obs_tau_reduction} and the Main Algorithm \ref{alg_solve}.
%where $a$ is an arbitrary node in $V$ and $\tau^a$ is the vector with $\tau^a(a)=k$ and $\tau^a(v)=0$ for $v\in V-a$.
Define the vector $\eta$ as follows:
 $\eta(v)=\Def_F(v)$ for $v\in V-a$, and $\eta(a)$ is the minimum nonnegative integer value for which
    \begin{equation}\label{eq:eta-a:choice}
        \min_{\substack{A\subsetneq V\\a\in A}} \bigl(|\rho(A)|+\eta(A)\bigr) \ge k.
    \end{equation}
Then  $\eta$ is a minimal $k$-half-extension of~$G$.
\end{theorem}
\begin{proof}[Proof of Theorems~\ref{th:approachA} and~\ref{th:approachB}]
The claim $\alpha_{\tt in}(G,k)\le \Def_F(V)$ in Theorem~\ref{th:approachA}
follows directly from Theorem~\ref{th:minmax:kforests}.
From now on, we assume that $\Def_F(V)>k$ in this theorem.

Let us define $a=\bot$ and $\tau={\bf 0}$ in the case of Theorem~\ref{th:approachA}, 
and $\tau=\tau^{a:k}$ in the case of Theorem~\ref{th:approachB}. Thus, in both cases we have $a\in V\cup\{\bot\}$,
and $F$ is an optimal solution of Problem~\ref{the_one_problem} for $(G,k,\tau)$.

First, we show that the vector $\eta$ is a $k$-half-extension, i.e.,\ it satisfies~\eqref{eq:half-extension}. Consider a set $A\subsetneq V$ with $A\ne \varnothing$.
Suppose that $a\notin A$.
We have $\indeg_F(A)\le |\rho(A)| + k (|A| - 1)$ and hence
\begin{equation*}
|\rho(A)|+\eta(A) = |\rho(A)| + \Def_F(A) = |\rho(A)| + k|A|-\indeg_F(A)
\ge k
\end{equation*}
If $a\in A$, then we are in the case of Theorem~\ref{th:approachB} and the claim holds since $\eta$ satisfies~\eqref{eq:eta-a:choice}.

Now let $\calS$ be the subpartition constructed in Lemma~\ref{lem_q} for an optimal solution $F$ for input $(G,k,\tau)$.
We claim that every $A\in\calS$ satisfies the following:
\begin{itemize}
\item $a\notin A$
(otherwise we are in the case of Theorem~\ref{th:approachB}, and $0<\Def_F(A)=k-\tau(A)-|\rho(A)|\le k-k-|\rho(A)|\le 0$ by eq.~\eqref{eq:Fclosed}
and the properties of $\calS$ - a contradiction).
\item $A\ne V$. If $a\ne \bot$, then this holds by the previous claim,
and if $a=\bot$, then this holds since $F_i$ is not a spanning tree for some $i$,
and $A$ is contained in a connected component of $F_i$.
\end{itemize}

Consider a node $v\in V-a$ with $\eta(v)=\Def_F(v)>0$,
and let $A$ be the set in $\calS$ containing $v$.
Eq.~\eqref{eq:Fclosed} gives $\eta(A)=\Def_F(A)=k-|\rho(A)|$
and hence $|\rho(A)|+\eta(A)=k$.
Thus, decreasing $\eta(v)$ will make~\eqref{eq:half-extension} false.
Decreasing $\eta(a)$ (in the case of Theorem~\ref{th:approachB}) will also make~\eqref{eq:half-extension} false
by the choice of $\eta(a)$. This shows that $\eta$ is a minimal $k$-half-extension.

It remains to observe that in Theorem~\ref{th:approachA} we have $\eta(V)=\Def_F(V)=\alpha_{\tt in}(G,k)$
by Theorem~\ref{th:minmax:kforests} and the fact that the subpartition $\calS$ is proper.
\end{proof}

Next, we discuss how to implement the approach in Theorem~\ref{th:approachB}. 
We will use Algorithm~\ref{alg:CompleteEta} to compute the minimum nonnegative integer value $\eta(a)$ for which~\eqref{eq:eta-a:choice} holds. The structure of Algorithm \ref{alg:CompleteEta} is inspired by Gabow's paper \cite{Gabow:STOC94}.
%This is can be done with the help of Theorem~\ref{th:complete-intersection}
%about complete $k$-intersections.

\begin{algorithm}[!htbp]
    \DontPrintSemicolon
    \SetNoFillComment
 
    let $G'$ be the graph obtained from $G$ by adding nodes $s,t$, edges $(s,v)$ with capacity $\eta(v)$ for $v\in V-a$,
    and an edge $(a,t)$ with capacity $k$ \\
    compute a maximum $s$-$t$ flow in $G'$; let $\bar G'$ be the residual graph and $f\in[0,k]$ be the value of the flow \\
    \If{the edge $(a,t)$ is saturated (i.e.,\ $f=k$) or some edge $(s,v)$ is not saturated}
    {
    	set $\eta(a)=k-f$
    }
    \Else
    {
	    let $\bar G$ be the graph obtained from $\bar G'$ by removing $s,t$ and incident edges \\
	    compute the maximum $\ell \in [0,k-f]$ such that the graph $\bar G^{\tt rev}$ has a complete $\ell$-intersection for node $a$\!\!\!\!\!\!\!\!\! \\
	    set $\eta(a)=k-f-\ell$
    }
    \caption{${\tt CompleteEta}(G,k,\eta)$ for a partial vector $\eta:V-a\rightarrow\{0,1,\ldots,k\}$}
    \label{alg:CompleteEta}
\end{algorithm}

\begin{theorem}
Algorithm~\ref{alg:CompleteEta} sets $\eta(a)$ to the minimum nonnegative integer value for which~\eqref{eq:eta-a:choice} holds.
It can be implemented in $O(km\log n)$ time.
\end{theorem}
\begin{proof}
Below, an {\em $s$-$t$ cut} is a set of nodes $U$ containing $t$ but not $s$, and its cost in a given
graph is the total cost of edges entering $U$.

Let $A^\ast$ be a set that minimizes the expression in~\eqref{eq:eta-a:choice},
and let $\theta = |\rho(A^\ast)|+\eta(A^\ast-a)$.
Clearly, $\eta(a)$ should be set to the value $\eta^\ast(a)\eqdef\max\{k-\theta,0\}$.
The set $A^\ast+t$ is an $s$-$t$ cut in $G'$ of cost $\theta$; therefore, $f\le \theta$.
In particular, if $f=k$, then $\theta\ge k$ and $\eta^\ast(a)=0$.
Let us assume that $f<k$. Then the edge $(a,t)$ is not saturated in a maximum flow
and hence we can assume w.l.o.g.\ that we chose $t=a$ when constructing $G'$,
and did not add the edge $(a,t)$. Also, $\bar G$ is obtained from $\bar G'$ by removing $s$ and its incident edges.
Note that $A^\ast$ is a minimum $s$-$a$ cut in $G'$ among cuts satisfying $A^\ast\ne V$.

Suppose that some edge $(s,v)$ is not saturated.
Then any minimum $s$-$a$ cut in $G'$ does not equal $V$ (since it does not contain $v$); thus, $A^\ast$ must be a minimum $s$-$a$ cut in $G'$.
The value $f$ equals the cost of $A^\ast$ in $G'$, which is $\theta$; therefore, $\eta^\ast(a)=\max\{k-f,0\}=k-f$, as desired.

Finally, let us assume that all edges $(s,v)$ are saturated.
In that case, we have $|\rho(A)|+\eta(A-a)=f+|\bar\rho(A)|$ for any $A\subseteq V$ containing $a$, where $\bar\rho(A)$ is the set of edges of $\bar G$ entering $A$. Therefore,
$$
\theta=f+\ell^\ast, \qquad\quad 
\ell^\ast=\min\{|\bar\rho(A)|\::\:A\subsetneq V,a\in A\}
    =\min\{|\bar\rho^{\tt rev}(\overline A)|\::\:\varnothing\ne\overline A\subseteq V-a\}
$$
By Edmonds' theorem about complete $k$-intersections, line 7 of Algorithm~\ref{alg:CompleteEta}
outputs $\ell=\min\{\ell^\ast,k-f\}=\min\{\theta-f,k-f\}$, and thus
$$
\eta(a)=k-f-\min\{\theta-f,k-f\}=k-\min\{\theta,k\}=\max\{k-\theta,0\}=\eta^\ast(a)
$$

It remains to discuss the complexity of Algorithm~\ref{alg:CompleteEta}.
The maximum flow in line 2 can be computed in time $O(km)$ since there are $f\le k$ augmentations
and each augmentation can be found in $O(m)$ time. The value $\ell$ in line 7 can be computed in $O(km\log n)$
time by the algorithm of Gabow~\cite{gabow}. This concludes the proof.
\end{proof}

\section{Algorithm for Problem~\texorpdfstring{\ref{undirected_problem_augmentation}}{C'}}\label{sec:undirected}

%\subsection{Undirected Edge Connectivity Augmentation Problem}\label{sec:1.3}
%Frank showed in~\cite{frank} that Problem~\ref{undirected_problem_augmentation}
%can be solved by the same approach as in the previous section:
%first construct an appropriate $(k,s)$-extension $G'$ of $G$ and
%then perform edge splittings in $G'$.
%In Section~\ref{sec:undirected} we show that $G'$ can be constructed in $O(k\delta m\log n)$ time. 
%For an undirected graph with $n$ nodes and $m'$ edges,
%a single vertex can be split in $O((k^2 n + m')\log n)$ time~\cite{edge_splitting}.
%$G'$ has $m'=O(m+\delta n)$ edges. Notice that $\delta m = \Omega(\delta k_G n) = \Omega(\delta (k - \delta) n)$, which is $\Omega(k n)$ for $\delta$ between 1 and $k - 1$. %If $\delta = k$, then again $\delta m = \Omega(k n)$, since $m = \Omega(n)$.
%We conclude that Problem~\ref{undirected_problem_augmentation}
%can be solved in $O(k\delta m\log n+(k^2 n + m + \delta n)\log n)= O(k(\delta m + k n)\log n) = O(k \delta m \log n)$ time.

The approach in the previous section can be easily adapted to solve
Problem~\ref{undirected_problem_augmentation}:
given an undirected graph $G$, find a smallest set of undirected edges whose addition to $G$ makes $G$  $k$-connected.
Let $\gamma(G,k)$ be the number of new edges. Below, we always assume that $k\ge 2$
(for $k=1$ the algorithm is trivial).

For $A\subseteq V$ let $\partial A$ be the set of edges between $A$ and $V-A$ in $G$.
We say that a vector $\eta:V\rightarrow \mathbb Z_{\ge 0}$ is a {\em $k$-extension of $G$}
if 
\begin{equation}\label{eq:k-extension}
|\partial A|+\eta(A)\ge k\qquad\quad \forall A\subsetneq V, A\ne \varnothing
\end{equation}
It is {\em minimal} if there is no $k$-extension $\eta'\ne\eta$ 
with $\eta'\le \eta$. % (i.e.,\ with $\eta'(v)\le \eta(v)$ for all $v\in V$).
Clearly, $\eta$ is a minimal $k$-extension of $G$
if and only if it is a minimal $k$-half-extension of $\vec G$
where $\vec G$ is a directed graph obtained from $G$ by replacing each edge $uv$
with two directed edges $(u,v),(v,u)$.
Thus, $\eta$ can be computed in $O(k\delta m\log n)$ time by the approach in the previous section.

To solve Problem~\ref{undirected_problem_augmentation}, we can now use the following algorithm.
\begin{enumerate}
\item Compute a minimal $k$-extension $\eta$ of $G$. If $\eta(V)$ is odd,
then pick an arbitrary $v\in V$ and increase $\eta(v)$ by 1.
Frank's proof of \cite[Lemma 4.2]{frank} shows that we now have $\eta(V)=2\gamma(G,k)$.
Note that $\eta(V)\le \delta n + 1$.
\item Construct an undirected graph $G'$ with $m'= m+\eta(V)$ edges by adding a node $s$ to $G$ and $\eta(v)$ edges $vs$ for each $v\in V$.
Clearly, $G'$ is $(k,s)$-connected, where $(k,s)$-connectivity is defined by applying Definition~\ref{def:ks-connected} to undirected graphs.
\item Repeat the following step while $s$ has incident edges in $G'=(V+s,E')$:
find a pair of edges $us,vs$ in $G'$ for which the graph $G''=(V+s,E'-\{us,vs\}+\{uv\})$ is $(k,s)$-connected, and replace $G'$ with $G''$.
The existence of such a pair was shown by Lov{\'a}sz~\cite{Lovasz:74,Lovasz:79}; see also \cite[Theorem 4.5]{frank}.
An efficient algorithm with complexity $O((k^2n+m')\log n)$ for performing all edge
splittings has been given in~\cite{edge_splitting}.
\end{enumerate}
The overall complexity of this algorithm is $O(k(\delta m + k n)\log n)$.
%
%For an undirected graph with $n$ nodes and $m'$ edges,
%a single vertex can be split in $O((k^2 n + m')\log n)$ time~\cite{edge_splitting}.
%$G'$ has $m'=O(m+\delta n)$ edges. 
Notice that $\delta m = \Omega(\delta k_G n) = \Omega(\delta (k - \delta) n)$, which is $\Omega(k n)$ for $\delta$ between 1 and $k - 1$. If $\delta = k$, then again $\delta m = \Omega(k n)$, since $m = \Omega(n)$.
We conclude that Problem~\ref{undirected_problem_augmentation}
can be solved in $O(k\delta m\log n+(k^2 n + m + \delta n)\log n)= O(k(\delta m + k n)\log n) = O(k \delta m \log n)$ time.

%Our algorithm does not improve the best known complexity of $\widetilde{O}(m)$ \cite{monte_carlo_augmentation} for the undirected augmentation problem. Our approach's only advantage for the undirected case in comparison to \cite{monte_carlo_augmentation} is derandomization.

We remark that in the undirected case, the min-max characterization for Problem~\ref{undirected_problem_augmentation} has been established by Cai and Sun~\cite{CaiSun:89};
a simplified proof can be found in~\cite{frank}.
%The same question can be asked for undirected graphs. In this case, one must minimize the number of new undirected edges to make the graph $k$-connected. Let us denote the size of the edge boundary of $X_i \subseteq V$ as $|\partial X_i|$. There is a min-max characterization for the undirected problem by Watanabe and Nakamura \cite{watanabe-nakamura}:

\begin{theorem}[\cite{CaiSun:89}]
%    Let $\gamma(G,k)$ be the minimum number of undirected edges that should be added to $G$    to make it $k$-connected. 
    If $k \geq 2$, then
    \begin{equation}
    \gamma(G,k)\;\;=\;\;\left\lceil \frac{1}{2} \;\;\max_{\raisebox{-5pt}{\mbox{\small $\calX$: a proper subpartition of $V$}}} \;\; \sum_{A\in\calX}(k-|\partial A|) \right\rceil
    \end{equation}
\end{theorem}

%%%%%%%%%%%%%%%%%%%%%%%%%%%%%%%%%%%%%%%%%%%%%%%%%%%%%%%%%%%%%%%%%%%%%%%%%%%%%%%%%%%%%%%%%%%%%%%%%%%%%%%%%%%%
%%%%%%%%%%%%%%%%%%%%%%%%%%%%%%%%%%%%%%%%%%%%%%%%%%%%%%%%%%%%%%%%%%%%%%%%%%%%%%%%%%%%%%%%%%%%%%%%%%%%%%%%%%%%
%%%%%%%%%%%%%%%%%%%%%%%%%%%%%%%%%%%%%%%%%%%%%%%%%%%%%%%%%%%%%%%%%%%%%%%%%%%%%%%%%%%%%%%%%%%%%%%%%%%%%%%%%%%%
%%%%%%%%%%%%%%%%%%%%%%%%%%%%%%%%%%%%%%%%%%%%%%%%%%%%%%%%%%%%%%%%%%%%%%%%%%%%%%%%%%%%%%%%%%%%%%%%%%%%%%%%%%%%
%%%%%%%%%%%%%%%%%%%%%%%%%%%%%%%%%%%%%%%%%%%%%%%%%%%%%%%%%%%%%%%%%%%%%%%%%%%%%%%%%%%%%%%%%%%%%%%%%%%%%%%%%%%%
%%%%%%%%%%%%%%%%%%%%%%%%%%%%%%%%%%%%%%%%%%%%%%%%%%%%%%%%%%%%%%%%%%%%%%%%%%%%%%%%%%%%%%%%%%%%%%%%%%%%%%%%%%%%
%%%%%%%%%%%%%%%%%%%%%%%%%%%%%%%%%%%%%%%%%%%%%%%%%%%%%%%%%%%%%%%%%%%%%%%%%%%%%%%%%%%%%%%%%%%%%%%%%%%%%%%%%%%%
%%%%%%%%%%%%%%%%%%%%%%%%%%%%%%%%%%%%%%%%%%%%%%%%%%%%%%%%%%%%%%%%%%%%%%%%%%%%%%%%%%%%%%%%%%%%%%%%%%%%%%%%%%%%
%%%%%%%%%%%%%%%%%%%%%%%%%%%%%%%%%%%%%%%%%%%%%%%%%%%%%%%%%%%%%%%%%%%%%%%%%%%%%%%%%%%%%%%%%%%%%%%%%%%%%%%%%%%%
%%%%%%%%%%%%%%%%%%%%%%%%%%%%%%%%%%%%%%%%%%%%%%%%%%%%%%%%%%%%%%%%%%%%%%%%%%%%%%%%%%%%%%%%%%%%%%%%%%%%%%%%%%%%
%%%%%%%%%%%%%%%%%%%%%%%%%%%%%%%%%%%%%%%%%%%%%%%%%%%%%%%%%%%%%%%%%%%%%%%%%%%%%%%%%%%%%%%%%%%%%%%%%%%%%%%%%%%%

\section*{AI disclosure}
We have used GPT-6 Astra for proofreading the paper. This helped us correct several claims that were false as stated. However, the content of the paper remained conceptually unchanged.

%%%%%%%%%%%%%%%%%%%%%%%%%%%%%%%%%%%%%%%%%%%%%%%%%%%%%%%%%%%%%%%%%%%%%%%%%%%%%%%%%%%%%%%%%%%%%%%%
%%%%%%%%%%%%%%%%%%%%%%%%%%%%%%%%%%%%%%%%%%%%%%%%%%%%%%%%%%%%%%%%%%%%%%%%%%%%%%%%%%%%%%%%%%%%%%%%
%%%%%%%%%%%%%%%%%%%%%%%%%%%%%%%%%%%%%%%%%%%%%%%%%%%%%%%%%%%%%%%%%%%%%%%%%%%%%%%%%%%%%%%%%%%%%%%%
%%%%%%%%%%%%%%%%%%%%%%%%%%%%%%%%%%%%%%%%%%%%%%%%%%%%%%%%%%%%%%%%%%%%%%%%%%%%%%%%%%%%%%%%%%%%%%%%
%%%%%%%%%%%%%%%%%%%%%%%%%%%%%%%%%%%%%%%%%%%%%%%%%%%%%%%%%%%%%%%%%%%%%%%%%%%%%%%%%%%%%%%%%%%%%%%%

%\begin{acks}
%To Robert, for the bagels and explaining CMYK and color spaces.
%\end{acks}

\bibliographystyle{alpha}
\bibliography{bibliography}

@inproceedings{gabow,
  title={A matroid approach to finding edge connectivity and packing arborescences},
  author={Gabow, Harold N},
  booktitle={Proceedings of the twenty-third annual ACM symposium on Theory of computing},
  pages={112--122},
  year={1991}
}

@book{schrijver-book:B,
  title={Combinatorial optimization: polyhedra and efficiency},
  author={Schrijver, Alexander},
  volume={24},
  year={2003},
  publisher={Springer Science \& Business Media}
}

@article{watanabe-nakamura,
  title={Edge-connectivity augmentation problems},
  author={Watanabe, Toshimasa and Nakamura, Akira},
  journal={Journal of Computer and System Sciences},
  volume={35},
  number={1},
  pages={96--144},
  year={1987},
  publisher={Elsevier}
}

@article{frank,
  title={Augmenting graphs to meet edge-connectivity requirements},
  author={Frank, Andr{\'a}s},
  journal={SIAM Journal on Discrete Mathematics},
  volume={5},
  number={1},
  pages={25--53},
  year={1992},
  publisher={SIAM}
}

@article{mader_splitting,
  title={Konstruktion aller n-fach kantenzusammenh{\"a}ngenden Digraphen},
  author={Mader, Wolfgang},
  journal={European Journal of Combinatorics},
  volume={3},
  number={1},
  pages={63--67},
  year={1982},
  publisher={Elsevier}
}

@inproceedings{monte_carlo_augmentation,
  title={Edge connectivity augmentation in near-linear time},
  author={Cen, Ruoxu and Li, Jason and Panigrahi, Debmalya},
  booktitle={Proceedings of the 54th Annual ACM SIGACT Symposium on Theory of Computing},
  pages={137--150},
  year={2022}
}

@inproceedings{edge_splitting,
  title={Fast edge splitting and Edmonds' arborescence construction for unweighted graphs},
  author={Bhalgat, Anand and Hariharan, Ramesh and Kavitha, Telikepalli and Panigrahi, Debmalya},
  booktitle={Proceedings of the nineteenth annual ACM-SIAM symposium on Discrete algorithms},
  pages={455--464},
  year={2008}
}

@article{union_find,
  title={Worst-case analysis of set union algorithms},
  author={Tarjan, Robert E and Van Leeuwen, Jan},
  journal={Journal of the ACM (JACM)},
  volume={31},
  number={2},
  pages={245--281},
  year={1984},
  publisher={ACM New York, NY, USA}
}

@article{CaiSun:89,
  title={The minimum augmentation of any graph to a {$K$}-edge-connected graph},
  author={Cai, Guo-Ray and Sun, Yu-Geng},
  journal={Networks},
  volume={19},
  number={1},
  pages={151--172},
  year={1989},
  publisher={Wiley Online Library}
}

@article{NaorGusfieldMartel:89,
  title={A fast algorithm for optimally increasing the edge connectivity},
  author={Naor, Dalit and Gusfield, Dan and Martel, Charles},
  journal={SIAM Journal on Computing},
  volume={26},
  number={4},
  pages={1139--1165},
  year={1997},
  publisher={SIAM}
}

@inproceedings{Gabow:STOC94,
  title={Efficient splitting off algorithms for graphs},
  author={Gabow, Harold N},
  booktitle={Proceedings of the twenty-sixth annual ACM symposium on Theory of Computing},
  pages={696--705},
  year={1994}
}

@inproceedings{NagamochiIbaraki:97,
  title={Deterministic {$\tilde O(nm)$} time edge-splitting in undirected graphs},
  author={Nagamochi, Hiroshi and Ibaraki, Toshihide},
  booktitle={Proceedings of the twenty-eighth annual ACM symposium on Theory of Computing},
  pages={64--73},
  year={1996}
}

@article{BenczurKarger:00,
  title={Augmenting undirected edge connectivity in {$\tilde O(n^2)$} time},
  author={Bencz{\'u}r, Andr{\'a}s A and Karger, David R},
  journal={Journal of Algorithms},
  volume={37},
  number={1},
  pages={2--36},
  year={2000},
  publisher={Elsevier}
}

@inproceedings{Gabow:STOC91,
  title={Applications of a poset representation to edge connectivity and graph rigidity},
  author={Gabow, Harold N},
  booktitle={Proceedings 32nd Annual Symposium of Foundations of Computer Science},
  pages={812--821},
  year={1991},
  organization={IEEE}
}

@inproceedings{Edmonds:69,
  title={Submodular functions, matroids, and certain polyhedra},
  author={Edmonds, Jack},
  booktitle={Combinatorial Optimization—Eureka, You Shrink! Papers Dedicated to Jack Edmonds 5th International Workshop Aussois, France, March 5--9, 2001 Revised Papers},
  pages={11--26},
  year={2003},
  organization={Springer}
}

@article{Edmonds:72,
  title={Edge-disjoint branchings},
  author={Edmonds, Jack},
  journal={Combinatorial algorithms},
  pages={91--96},
  year={1973},
  publisher={Academic Press}
}

@book{Lovasz:74,
  author = {L{\'a}szl{\'o} Lov{\'a}sz},
  title = {Conference on Graph Theory},
  publisher = {Lecture, Prague},
  year = 1974
}

@book{Lovasz:79,
  title={Combinatorial problems and exercises},
  author={Lov{\'a}sz, L{\'a}szl{\'o}},
  volume={361},
  year={2007},
  publisher={American Mathematical Soc.}
}

@inproceedings{GabowWestermann:92,
  title={Forests, frames, and games: algorithms for matroid sums and applications},
  author={Gabow, Harold and Westermann, Herbert},
  booktitle={Proceedings of the twentieth annual ACM symposium on Theory of computing},
  pages={407--421},
  year={1988}
}

@inproceedings{blikstad,
  title={Fast algorithms via dynamic-oracle matroids},
  author={Blikstad, Joakim and Mukhopadhyay, Sagnik and Nanongkai, Danupon and Tu, Ta-Wei},
  booktitle={Proceedings of the 55th Annual ACM Symposium on Theory of Computing},
  pages={1229--1242},
  year={2023}
}

@article{terao,
  title={Faster matroid partition algorithms},
  author={Terao, Tatsuya},
  journal={ACM Transactions on Algorithms},
  volume={21},
  number={2},
  pages={1--26},
  year={2025},
  publisher={ACM New York, NY}
}

@inproceedings{quanrud,
  title={Faster exact and approximation algorithms for packing and covering matroids via push-relabel},
  author={Quanrud, Kent},
  booktitle={Proceedings of the 2024 Annual ACM-SIAM Symposium on Discrete Algorithms (SODA)},
  pages={2305--2336},
  year={2024},
  organization={SIAM}
}

@article{almost_linear_maxflow,
  title={Maximum flow and minimum-cost flow in almost-linear time},
  author={Chen, Li and Kyng, Rasmus and Liu, Yang and Peng, Richard and Probst Gutenberg, Maximilian and Sachdeva, Sushant},
  journal={Journal of the ACM},
  volume={72},
  number={3},
  pages={1--103},
  year={2025},
  publisher={ACM New York, NY}
}

@article{Imai:83,
    author = {Hiroshi Imai},
    title = {Network-flow algorithms for lower-truncated transversal polymatroids},
    journal = {Journal of the Operations Research Society of Japan},
    volume = {26},
    number = {3},
    pages = {186--210},
    year = {1983},
}

@article{RoskindTarjan:85,
    author = {James Roskind and Robert E. Tarjan},
    title = {A note on finding minimum-cost edge-disjoint spanning trees},
    journal = {Mathematics of Operations Research},
    volume = {10},
    number = {4},
    pages = {701--708},
    year = {1985},
}

@inproceedings{GabowStallman:85,
      title={Efficient algorithms for the parity and intersection problems}, 
      author={Harold N. Gabow and Matthias Stallmann},
      year=1985,
      booktitle = {Automata, Languages, and Programming, 12th Colloquium, Lecture Notes in Computer Science},
      pages = {210--220},
      volume = {194},
      editor = {W. Brauer},
      publisher = {Springer-Verlag, New York}
}

@article{KK,
  title={Maximally distant trees and principal partition of a linear graph},
  author={Kishi, Genya and Kajitani, Yoji},
  journal={IEEE Transactions on Circuit Theory},
  volume={16},
  number={3},
  pages={323--330},
  year={1969},
  publisher={IEEE}
}

@article{Wh,
  title={The union of matroids and the rigidity of frameworks},
  author={Whiteley, Walter},
  journal={SIAM Journal on Discrete Mathematics},
  volume={1},
  number={2},
  pages={237--255},
  year={1988},
  publisher={SIAM}
}

@article{WW,
  title={The algebraic geometry of motions of bar-and-body frameworks},
  author={White, Neil and Whiteley, Walter},
  journal={SIAM Journal on Algebraic Discrete Methods},
  volume={8},
  number={1},
  pages={1--32},
  year={1987},
  publisher={SIAM}
}

@article{OIW,
  title={Topological degrees of freedom and mixed analysis of electrical networks},
  author={Ohtsuki, Tatsuo and Ishizaki, Yasutoshi and Watanabe, Hitoshi},
  journal={IEEE Transactions on Circuit Theory},
  volume={17},
  number={4},
  pages={491--499},
  year={1970},
  publisher={IEEE}
}

@article{IF,
  title={Use of matroid theory in operations research, circuits and systems theory},
  author={Iri, Masao and Fujishige, Satoru},
  journal={International Journal of Systems Science},
  volume={12},
  number={1},
  pages={27--54},
  year={1981},
  publisher={Taylor \& Francis}
}

@article{Tay,
  title={Rigidity of multi-graphs. I. Linking rigid bodies in {$n$}-space},
  author={Tay, Tiong-Seng},
  journal={Journal of Combinatorial Theory, Series B},
  volume={36},
  number={1},
  pages={95--112},
  year={1984},
  publisher={Elsevier}
}

@article{gabow_polymatroids,
  title={Algorithms for graphic polymatroids and parametrics-sets},
  author={Gabow, Harold N},
  journal={Journal of Algorithms},
  volume={26},
  number={1},
  pages={48--86},
  year={1998},
  publisher={Elsevier}
}

@article{routing,
  title={On the resiliency of static forwarding tables},
  author={Chiesa, Marco and Nikolaevskiy, Ilya and Mitrovi{\'c}, Slobodan and Gurtov, Andrei and Madry, Aleksander and Schapira, Michael and Shenker, Scott},
  journal={IEEE/ACM Transactions on Networking},
  volume={25},
  number={2},
  pages={1133--1146},
  year={2016},
  publisher={IEEE}
}

@article{NagamochiIbaraki_sparsification,
  title={A linear-time algorithm for finding a sparse k-connected spanning subgraph of a k-connected graph},
  author={Nagamochi, Hiroshi and Ibaraki, Toshihide},
  journal={Algorithmica},
  volume={7},
  number={1},
  pages={583--596},
  year={1992},
  publisher={Springer}
}

@article{Tutte,
  title={On the problem of decomposing a graph into {$n$} connected factors},
  author={William T. Tutte},
  journal={Journal of The London Mathematical Society-second Series},
  year={1961},
  pages={221-230},
}

@article{flowrounding,
  title={Flow rounding},
  author={Kang, Donggu and Payor, James},
  journal={arXiv preprint arXiv:1507.08139},
  year={2015}
}

@inproceedings{deterministic-max-flow,
  title={A deterministic almost-linear time algorithm for minimum-cost flow},
  author={Van Den Brand, Jan and Chen, Li and Kyng, Rasmus and Liu, Yang P and Peng, Richard and Gutenberg, Maximilian Probst and Sachdeva, Sushant and Sidford, Aaron},
  booktitle={2023 IEEE 64th Annual Symposium on Foundations of Computer Science (FOCS)},
  pages={503--514},
  year={2023},
  organization={IEEE}
}


\appendix
% \renewcommand{\thesection}{Appendix \Alph{section}}

%\appendix

\section{%Appendix A \quad 
Proof of Theorem~\ref{th:proc_augment}} \label{sec:proc_augment:proof}
Consider part (a).
The first claim can be verified by comparing how many times $i$ occurs in $(\ell(e_1),\ldots,\ell(e_r))$ before and after the update.
Let us show the second claim. By the definition of the auxiliary graph, if $\ell_{i+1}=\varnothing$ for some $i\in[r-1]$ (i.e.,\ if $e_{i+1}\notin F$)
then $\head(e_{i})=\head(e_{i+1})$.
Suppose that $\ell_{i}=\ell_{i+1}=\varnothing$ for some $i\in[r-1]$.
Let us remove $e_i$ from $P$.
The new path $P'$ is a valid $V$-$t$ path
(if $i>1$, then this holds since $\head(e_{i-1})=\head(e_{i})=\head(e_{i+1})$,
and hence $(e_{i-1},e_{i+1})$ is in the auxiliary graph). Furthermore, {\tt Augment}($P$)
and {\tt Augment}($P'$) yield the same result. By repeatedly applying this operation,
we can assume w.l.o.g.\ that $(\ell_1,\ell_2,\ldots,\ell_r)$
does not contain two consecutive $\varnothing$'s.
Augmenting $P$ changes the set of edges in $F$ as follows:
$e_1$ is added to $F$, and for every $i\in[2,r-1]$ with $\ell_{i+1}=\varnothing$,
$e_i$ is removed from $F$ and $e_{i+1}$ is added to $F$. The claim can now be easily verified.

Now let us prove parts (b,c). Unless noted otherwise, in this section $F=F_1\cup \ldots \cup F_k\in\calG^k\cap\calD$ always denotes a feasible solution of Problem~\ref{the_one_problem},
$e,f,x,y$ denote elements of $E$, and $s,t$ denote special elements that are not in $E$.

Below, we will define several directed graphs. For all of them, the set of nodes will be $E\sqcup\{s,t\}$,
so we will view these graphs as subsets of edges. For such a graph $D$, we define $D^{\tt inv}$
to be the graph obtained from $D$ by reversing edge orientations and additionally swapping nodes $s$ and $t$.
For example, $(s,e)\in D$ if and only if $(e,t)\in D^{\tt inv}$.

For a matroid $\calM$ on $E$ and an independent set $X\in\calG$, define a graph $D_\calM(X)$ via
\begin{eqnarray*}
D_\calM(X) &=&\{(x,y)\::x\notin X,y\in X,X+x-y\in\calM\} 
\;\cup\;
\{(x,t)\::x\notin X,X+x\in \calM\}  
\end{eqnarray*}
Using this notation, we now define the following graphs:
\begin{eqnarray*}
D^{\tt union}(F)&=&D_\calG(F_1)\cup\ldots\cup D_\calG(F_k) \\
D^{\tt intersection}(F)&=& D_{\calG^k}(F) \cup [D_{\calD}(F)]^{\tt inv}
\end{eqnarray*}
%Note that $D^{\tt union}(F)\subseteq D(F)$. 
It can be checked that
\begin{eqnarray*}
D(F)\;\;=\;\; D^{\tt union}(F) &\cup& \{(v,e)\::\:\index_F(v)<k-\tau(v),\head(e)=v,e\notin F\}  \\
&\cup & \{(e,f)\::\:e\in F,f\notin F,\head(e)=\head(f)\}
\end{eqnarray*}

The following are classical results about matroid union and intersection problems; see, e.g.,~\cite{schrijver-book:B}.
%In particular, we will use the followings facts.
\begin{theorem}[{\cite[Theorem 42.4]{schrijver-book:B}}] \label{th:matroid-union}
For any $F\in\calG^k$ and $f\notin F$, the following holds: $F+f\in\calG^k$ if and only if $D^{\tt union}(F)$ has an $f$-$t$ path.
\end{theorem}

\begin{theorem}[{\cite[Theorems 41.2-41.3]{schrijver-book:B}}] \label{th:matroid-intersection}
 $F\in\calD\cap\calG^k$ is not optimal if and only if $D^{\tt intersection}(F)$ has an $s$-$t$ path.
\end{theorem}

%\begin{lemma}\label{lemma:intersection:xy}
%If $(x,y)\in D^{\tt intersection}(F)$ with $x\notin F$, $y\in F$ and $F-y+x\in\calG^k$ then
% there is either  an  $x$-$y$ path or and $x$-$t$ path in $D(F)$. 
%\end{lemma}
%\begin{proof}
%By Theorem~\ref{th:matroid-union}, graph $D^{\tt union}(F-y)$ contains an $x$-$t$ path $P$.
%The claim will now follow from the two facts below.
%\begin{itemize}
%\item If $(f,e)\in P$, $y\ne f$ then either $(f,e)\in D(F)$ or $(f,y)\in D(F)$.
%Indeed, we have $f\notin F_i-y$, $e\in F_i-y$ %(implying $f\notin F_i$, $e\in F_i$) 
%and $(F_i-y)-e+f\in \calG$ for some~$i\in[k]$.
%Assume that $y\in F_i$ (otherwise $F_i-e+f\in \calG$ and hence $(f,e)\in D(F)$).
%Applying the matroid exchange axiom to sets $F_i-y-e+f$ and $F_i$
%yields that either $F_i-e+f\in \calG$ (in which case $(f,e)\in D(F)$)
%or $F_i-y+f\in \calG$ (in which case $(f,y)\in D(F)$).
%\item If $(e,t)\in P$ then either $(e,y)\in D(F)$ or $(e,t)\in D(F)$. Indeed, we have $e\notin F_i-y$, $(F_i-y)+e\in\calG$ for some $i\in [k]$.
%If $y\in F_i$ then $(e,y)\in D(F)$, otherwise $F_i+e\in\calG$ and hence $(e,t)\in D(F)$.
%\end{itemize}
%\end{proof}

From these facts, we can infer Theorem~\ref{th:proc_augment}(c) as follows.
\begin{lemma}
If $F\in\calD\cap\calG^k$ is not optimal, then $D(F)$ has a $V$-$t$ path.
\end{lemma}
\begin{proof}
By Theorem~\ref{th:matroid-intersection}, $D^{\tt intersection}(F)$ has an $s$-$t$ path $P$.
Consider a shortest such path.
It may contain edges of the following types.
\begin{itemize}
\item $(s,e)$ with $e\notin F$ and $F+e\in \calD$. Let $v=\head(e)$. Then $\indeg_F(v)<k-\tau(v)$, and hence $(v,e)\in D(F)$.
\item $(e,f)$ with $f\notin F$, $e\in F$ and $F-e+f\in\calD$.
Since $P$ is shortest, we have $(s,f)\notin D^{\tt intersection}(F)$ and thus $F+f\notin \calD$.
This implies that $\head(e)=\head(f)$, and hence $(e,f)\in D(F)$.
%\item $(f,e)$ as in Lemma~\ref{lemma:intersection:xy}. Then $D(F)$ contains either an $f$-$e$ path or a $f$-$t$ path.
\item $(x,y)$ with $x\notin F$, $y\in F$ and $F-y+x\in\calG^k$. 
We claim that $D(F)$ contains either an $x$-$y$ path or an $x$-$t$ path.
Indeed, the graph $D^{\tt union}(F-y)$ contains an $x$-$t$ path $P'$  (by Theorem~\ref{th:matroid-union}).
The claim will now follow from the two facts below.
\begin{itemize}
\item[$-$] If $(f,e)\in P'$, $y\ne f$, then either $(f,e)\in D(F)$ or $(f,y)\in D(F)$.
Indeed, we have $f\notin F_i-y$, $e\in F_i-y$ %(implying $f\notin F_i$, $e\in F_i$) 
and $(F_i-y)-e+f\in \calG$ for some~$i\in[k]$.
Assume that $y\in F_i$ (otherwise $F_i-e+f\in \calG$ and hence $(f,e)\in D(F)$).
%Also, $F_i\in\calG$. 
Applying the matroid exchange axiom to the sets $F_i-y-e+f$ and $F_i$
yields that either $F_i-e+f\in \calG$ (in which case $(f,e)\in D(F)$)
or $F_i-y+f\in \calG$ (in which case $(f,y)\in D(F)$).
\item[$-$] If $(e,t)\in P'$, then either $(e,y)\in D(F)$ or $(e,t)\in D(F)$. Indeed, we have $e\notin F_i-y$, $(F_i-y)+e\in\calG$ for some $i\in [k]$.
If $y\in F_i$, then $(e,y)\in D(F)$; otherwise, $F_i+e\in\calG$ and hence $(e,t)\in D(F)$.
\end{itemize}
\item $(e,t)$ with $e\notin F$ and $F+e\in \calG^k$.
Then the graph $D^{\tt union}(F)\subseteq D(F)$ contains an $e$-$t$ path (by Theorem~\ref{th:matroid-union}).
\end{itemize}
From these facts, we can conclude that $D(F)$ contains a $V$-$t$ path.
\end{proof}

It remains to prove Theorem~\ref{th:proc_augment}(b).
For that, we introduce the following definition:
\begin{definition}
Given an independent set $X\in\calG$, an {\em $X$-path} is a sequence of distinct elements $Q=(x_1,y_1,\ldots,x_r,y_r)$
where $r\ge 1$, $y_r\in E\cup\{t\}$, all other elements of $Q$ belong to $E$, and the following holds:
%satisfying the following properties:
\begin{itemize}
%\item[(1)] $y_r\in E\cup\{t\}$, and all other elements of $P$ belong to $E$.
\item[(1)] $(x_i,y_i)\in D_\calG(X)$ for $i\in [r]$.
\item[(2)] $(x_i,t)\notin D_\calG(X)$ for all $i\in[r]$ with $y_i\ne t$.
\item[(3)] $(x_i,y_j)\notin D_\calG(X)$ for all $i,j\in[r]$ with $i<j$.
\end{itemize}
For such $Q$, we define $X\oplus Q=X+\{x_1,\ldots,x_r\}-\{y_1,\ldots,y_r\}\subseteq E$.
\end{definition}
It can be seen that if $P$ is a $V$-$t$ path in $D(F)$ without shortcuts,
then the operation ${\tt Augment}(P)$ changes the forest $F_i$ to $F_i\oplus Q_i$ for some $F_i$-path $Q_i$.
Thus, Theorem~\ref{th:proc_augment}(b) will follow from the result below.
%\begin{lemma}
%Let $P$ be a $V$-$t$ path in $D(F)$ without shortcuts.
%Then ${\tt Augment}(P)$ produces feasible forests $F'=F'_1\sqcup\ldots\sqcup F'_k$.
%\end{lemma}
\begin{lemma}
If $Q=(x_1,y_1,\ldots,x_r,y_r)$ is an $X$-path, then $X+\{x_1,\ldots,x_r\}-\{y_1,\ldots,y_r\}\in\calG$.
\end{lemma}
\begin{proof}
We use induction on $r$. For $r=1$, the claim follows directly from the definitions; suppose that $r\ge 2$.
Denote $X'=X+x_1-y_1\in\calG$. It suffices to show that $Q'=(x_2,y_2,\ldots,x_r,y_r)$ is an $X'$-path;
the claim will then follow from the induction hypothesis. 

Consider $i\in[2,r]$. Applying the matroid exchange axiom to the sets $X+x_1-y_1-y_i\subseteq X+x_1-y_1$ and $X+x_i-y_i$
(that are both in $\calG$) gives that either $X+x_1-y_i\in\calG$ or $X+x_1-y_1+x_i-y_i\in\calG$.
The former is impossible since $(x_1,y_{i})\notin D_\calG(X)$; hence $X'+x_i-y_i\in\calG$, i.e.,\ $(x_i,y_i)\in D_\calG(X')$.

Suppose that $(x_i,t)\in D_\calG(X')$ for some $i\in[2,r]$ with $y_i\ne t$,
i.e.,\ $X'+x_i\in\calG$. Applying the matroid exchange axiom to the sets
$X$ and $X+x_1-y_1+x_i$ gives that either $X+x_1\in \calG$ or $X+x_i\in\calG$.
We thus have either $(x_1,t)\in D_\calG(X)$ or $(x_i,t)\in D_\calG(X)$,
either of which contradicts condition (2).

Finally, suppose that $(x_i,y_j)\in D_\calG(X')$ for some $i,j\in[2,r]$ with $i<j$.
As shown above, we must have $y_j\ne t$. Applying the matroid exchange axiom to the sets
$X-y_j$ and $X+x_1-y_1+x_i-y_j$ gives that either $X+x_1-y_j\in\calG$
or $X+x_i-y_j\in\calG$. We thus have either $(x_1,y_j)\in D_\calG(X)$
or $(x_i,y_j)\in D_\calG(X)$, either of which contradicts condition (3).
\end{proof}

\end{document}